\documentclass{article}

\usepackage[final]{neurips_2024}

\usepackage{comment}
\usepackage[utf8]{inputenc} % allow utf-8 input
\usepackage[T1]{fontenc}    % use 8-bit T1 fonts
\usepackage{hyperref}       % hyperlinks
\usepackage{url}            % simple URL typesetting
\usepackage{booktabs}       % professional-quality tables
\usepackage{amsfonts}       % blackboard math symbols
\usepackage{nicefrac}       % compact symbols for 1/2, etc.
\usepackage{microtype}      % microtypography
\usepackage{xcolor}         % colors

\usepackage{comment}
\usepackage{nicefrac}
\usepackage{booktabs} % For formal tables
\usepackage[ruled]{algorithm2e} % For algorithms

\usepackage{amsthm}
\usepackage{tikz}
\usepackage{times}
\usepackage{amsmath}
\usepackage{amsthm}
\usepackage{amssymb}
\usepackage{bbm}
\usepackage{cleveref}

\usepackage{algorithmic}
\newtheorem{theorem}{Theorem}[section]
\newtheorem{lemma}[theorem]{Lemma}

\newtheorem{proposition}[theorem]{Proposition}
\newtheorem{corollary}[theorem]{Corollary}

\newtheorem{definition}[theorem]{Definition}

\usepackage{enumerate}

\newtheorem{claim}{Claim}
 \newtheorem{remark}{Remark}
\renewcommand{\Pr}{\mathop{\bf Pr\/}}
\newcommand{\E}{\mathop{\bf E\/}}

\newcommand{\SP}{\mathrm{SP}}

\newcommand{\rev}{\mathrm{Rev}}

\newcommand{\reg}{\mathrm{Reg}}

\newcommand{\nats}{\mathbb N}

\newcommand{\eps}{\epsilon}

\newcommand{\calE}{\mathcal{E}}

\def\wt{\widetilde}
\def\wh{\widehat}
\def\eps{\varepsilon}

\def\tildeb{\smash{\tilde{b}}}

\newcommand{\nf}{\nicefrac}

\setcitestyle{authoryear}

\title{Randomized Truthful Auctions with Learning Agents}

\author{%
 Gagan Aggarwal \\
  Google Research\\
  \texttt{gagana@google.com} \\
  \And
  Anupam Gupta \\
  New York University, Google Research\\
  \texttt{anupam.g@nyu.edu} \\
  \And
  Andres Perlroth \\
  Google Research\\
  \texttt{perlroth@google.com} \\
  \And
  Grigoris Velegkas\thanks{Part of the work was done while the author was a research intern at Google Research in Mountain View.} \\
  Yale University\\
  \texttt{grigoris.velegkas@yale.edu} \\
}

\begin{document}

\maketitle

\begin{abstract}
 We study a setting where agents use no-regret learning algorithms to participate in repeated auctions. \citet{kolumbus2022auctions} showed, rather surprisingly, that when bidders participate in second-price auctions using no-regret bidding algorithms, no matter how large the number of interactions $T$ is, the runner-up bidder may not converge to bidding truthfully. Our first result shows that this holds for \emph{general deterministic} truthful auctions. We also show that the ratio of the learning rates of the bidders can \emph{qualitatively} affect the convergence of the bidders. 
Next, we consider the problem of revenue maximization in this environment. In the setting with fully rational bidders,  \citet{myerson1981optimal} showed
that revenue can be maximized by using a second-price auction with reserves.
We show that, in stark contrast, in our setting with learning bidders, \emph{randomized} auctions
   can have strictly better revenue guarantees than second-price
    auctions with reserves, when $T$ is large enough. 
    %To do this, we provide a black-box transformation from any truthful
    %auction $A$ to an auction $A'$ such that: i) all mean-based no-regret learners
    %that participate in $A'$ converge to bidding truthfully, ii)
    %the distance between the allocation rule and the payment
    %rule between $A,A'$ is negligible.    
    Finally, we study revenue maximization in the non-asymptotic regime.
We define a notion of {\em auctioneer regret} comparing the revenue generated to the revenue of a second price auction with truthful bids. When the auctioneer has to use
    the same auction throughout the interaction,
    we show an (almost) tight regret bound 
    of $\smash{\widetilde \Theta(T^{3/4})}.$ If the auctioneer can  change auctions during
    the interaction, but in a way that is oblivious
    to the bids, we show an (almost) tight
    bound of $\smash{\widetilde \Theta(\sqrt{T})}.$ 
    
\end{abstract}

\section{Introduction}\label{sec:intro}
In auction design, truthfulness is a highly sought-after property. It allows bidders to simply reveal their true valuations, simplifying the bidding process. In the standard single item setting with fully rational profit-maximizing bidders, Myerson's seminal paper~\cite{myerson1981optimal} shows that an auctioneer can achieve optimal revenue by using a {\em truthful and deterministic} auction mechanism -- a Second Price Auction (SPA) with a reserve price. 
%In auction design, truthfulness is a highly sought-after property. It allows bidders to simply reveal their true valuations, streamlining the bidding process. Surprisingly, \citet{myerson1981optimal} shows that with fully rational profit-maximizing bidders, truthful mechanisms can achieve optimal revenue for the auctioneer by implementing a Second Price Auction (SPA) with a reserve price.

%\ganote{We are trying to draw a contrast between the fully rational bidders, who are happy with a deterministic truthful auction, vs learning bidders where we need randomness to get a truthful auction with good properties. Attempted to rewrite the para a bit to help that narrative. Please feel free to edit}\andresnote{i think i wanted to also emphasize that truthfulness in itself is a very desirable property and moreover you can implement the opt mechanism w/o loosing this property. this is important because we are restricting to truthful auctions in this paper.}
%\andresnote{split paragraph?}

In many applications nowadays, buyers no longer bid directly in the auction but, instead, use learning algorithms to bid on their behalf. For example, in online advertising, platforms offer automated bidding tools that manage ad campaigns on behalf of advertisers.
Such bidders learn to bid over many rounds and are not fully rational. In a surprising result, \citet{kolumbus2022auctions} show that some appealing properties of second-price auctions break down in the presence of such learning bidders. In particular, when (profit-maximizing) bidders use no-regret learning algorithms, the second-price auction does not achieve as much revenue as with fully rational bidders. Indeed, bidders do not learn to bid their value, and consequently, the runner-up bidder’s bid is less than their value with positive probability, which diminishes the second price auction’s revenue. Moreover, \citet{kolumbus2022and} show that for a setting where rational agents are using learning algorithms to bid, then it is no longer optimal to truthfully submit their value as the input to the learning algorithm. This raises a crucial question: are there truthful auctions that promote convergence to the true valuations within a learning environment, and can they also guarantee strong revenue performance?

In this paper we provide an affirmative answer to this question.
%\grigorisnote{I changed it to singular, since we have only mentioned one question so far -- feel free to revert it to ``questions'' if you feel that it's more appropriate}. 
{In doing so, we also showcase the value of \emph{randomized} mechanisms} --- often overlooked in settings with profit-maximizing bidders --- for environments where bidders are learning agents. While randomization introduces inherent inefficiencies due to allocations to low-valuation bidders, this very behavior facilitates learning among low-valuation bidders. A revenue-maximizing auctioneer must now carefully balance the randomization within a truthful mechanism to incentivize learning without incurring excessive revenue loss due to mis-allocation.

We build our theory based on the model presented by \citet{kolumbus2022auctions}. We consider single-item repeated interactions over $T$ periods. There are two profit-maximizing bidders participating in the auctions, %\footnote{We discuss in \Cref{??} how our results apply to multiple bidders.} 
with valuations that are drawn independently from the same distribution, and fully persistent over time. {This
assumption is motivated by online ad auctions, where multiple auctions are taking place every second, and the valuations of the advertisers remain stable for certain time scales, e.g., a day or a week. Thus, there is typically a very large sequence of auctions where the valuations of the participating agents are persistent.}
% \footnote{In the paper, we also discuss how our results also apply to a model where valuations are chosen adversarially.} 
Bidders use mean-based no-regret learning algorithms \citep{braverman2018selling} and receive full feedback on which they base their updates. {(Many of our results extend immediately to multiple bidders. We discuss other extensions, such as the partial feedback settings, in
\Cref{app:extensions}.) }
%\footnote{We discuss in \Cref{??} how our results would change in a partial feedback setting.} \anupamnote{Should we just have footnote saying: "We discuss the extensions to multiple bidders and to partial feedback settings in \S blah."}
The auctioneer focuses on truthful auctions, and their objective is to maximize the total revenue they achieve over the $T$ rounds of interaction. Our results are the following:

\subsection{Our Results and Techniques}

\noindent\textbf{Limitations of Deterministic Auctions.} Our first set of results (in \Cref{sec:deterministic auctions}) characterize
the convergence of learners who are using Multiplicative Weights Update (MWU) in repeated \emph{deterministic auctions}. 
In particular, we show the following sharp phase transition: 
\begin{itemize}
\item If the \emph{learning rate} of the winning type is at least
as fast as the learning rate of the runner-up type, then
the runner-up type will not 
converge to bidding truthfully, 
even as $T \rightarrow \infty$; in fact, 
it will be bidding strictly below its true value, in expectation.
\item On the other hand, we show that
in many auctions, such as SPA, if the learning
rate of the runner-up type is strictly faster than that of the winning type, then the runner-up type will indeed converge to truthful bidding.
\end{itemize}
These generalize the results of \citet{kolumbus2022auctions}
who showed that in SPA, when bidders are using MWU with the same learning
rate, then the low type will not converge to bidding truthfully. The main challenges to proving this set of  results arise from our study of general
deterministic auctions, which have less structure
than second-price auctions. Indeed, small differences
in the learning rates can affect the landscape qualitatively, 
as is manifested from our results. Moreover, while the auctions are deterministic, the learning algorithms are randomized and highly correlated. Hence our 
approach is to break down the interaction into several
epochs and establish some qualitative properties which hold,
with high probability, at the end of each epoch. This requires
a careful accounting of the cumulative utility of each bid
of both bidders within
every epoch; in particular, if our estimation is off by even 
some $\omega(1)$ term, then it will not be sufficient to
establish our result. 

\noindent\textbf{Strictly-IC Auctions and the Power of Randomized Mechanisms.}
{The results in \Cref{sec:deterministic auctions} show that since the low valuation bidder tends to underbid, an auctioneer using SPA with reserve makes strictly less revenue than that predicted by the model with rational agents.}
%The results in \Cref{sec:deterministic auctions} show that since the low valuation bidder tends to underbid, the auctioneer’s revenue is strictly less than that predicted by the model with rational agents.
%\andresnote{mention this discussion about the learning rate? shall we mention that also provide quantitive description of the limit bid distribution?} \andresnote{split para?}\grigorisnote{I've split the paragraph and added the discussion. I think the qualitative result will not add much, and it is the same one that was shown from Kolumbus and Nisan, so maybe it is ok not to talk about it? }
%Given this limitation of deterministic mechanisms, 
Motivated by this, we consider a special class of randomized auctions called \emph{strictly-IC auctions}. These are randomized truthful auctions where for each bidder, it is strictly better to bid their true valuation compared to any other bid. We show that any strictly-IC auction is asymptotically truthful: 
that is, the limit point of the bidder's bid converges to their true value. Furthermore, we provide a black-box transformation from any truthful auction $A$ (deterministic or not) to a randomized auction $A'$ that has the following two properties: (i) the bidders converge towards truthful bidding, and (ii) the difference between the allocation and payment rules of the original auction $A$ and its strictly-IC counterpart $A'$ are negligible for any bid profile. Hence, such an auction $A'$ behaves similarly to $A$, but, crucially, it conveys information to the low bidder to help it converge to truthful bidding.
{As a corollary of this result, we get that SPA with reserve is not revenue-maximizing in this setting, and that randomization can get strictly more revenue than SPA with reserve. This is in stark contrast with the seminal result of \citet{myerson1981optimal} which shows that SPA with reserve is optimal for rational bidders.} 
%\grigorisnote{TODO. might need to  edit the last sentence and only mention SPA with reserve.} 

% \anupamnote{Added a version of the claim from the EC rebuttal here. Is this a good place, or should we move it?}
{{At a more conceptual level, }our results for randomized mechanisms can be viewed as showing that having enough randomness is key to the low bidder converging to truthful bidding: this randomness can come from the process itself, e.g., if bidder values are independently drawn in each round, as in \citet{feng2021convergence}. But if not, and if the ranking of the bidders does not change much due to the lack of inherent randomness, our results show that injecting external randomness into the auction
induces the desired learning behavior and hence improves the revenue.
{Having persistent valuations is just one case of the ranking of the bidders remaining stable over time: studying this case allows us to showcase our main ideas, but a central message of our work is that the presence or absence of stability in the rankings of the bidders is the main factor that dictates convergence to truthful bidding.}
}

\noindent\textbf{A Non-Asymptotic Analysis.}
Our next set of results in \Cref{sec:no-regret-learning-auctioneer} address the non-asymptotic regime. Here we consider the \emph{prior-free} setting,
meaning that the valuations of the bidders could 
be drawn from potentially different distributions that are unknown to the auctioneer. In order to evaluate its revenue performance when bidders are learning agents, we introduce the notion of \emph{auctioneer regret} for an auction, 
which measures the difference between the revenue achieved over $T$ rounds of implementing a given auction with learning bidders and the revenue achieved by implementing the optimal auction with rational bidders (i.e., SPA with a reserve price). {\Cref{prop:lower-bound-constant-auctions}} shows that if the auctioneer is constrained to use the same auction rule for all $T$ rounds, then no truthful auction --- deterministic or randomized --- can achieve an auctioneer-regret better than $\widetilde{O}(T^{3/4})$ in the setting of adversarial valuations.
However, if the auctioneer can change the auction rule just once within the $T$ rounds, with the change happening at a time independent of the bid history, then the auctioneer's regret drops to  $\widetilde{O}(\sqrt{T})$, as we show in {\Cref{sec:no-regret-learning-auctioneer}} Moreover, we show in \Cref{prop:root-T-lower-bound-all-environments} that this bound of $\widetilde{O}(\sqrt{T})$ is optimal even
if the auctioneer can design the auction schedule.
% conditioned on the realization of the types. and we design a two-stage mechanism (which uses a randomized auction in the first stage, followed by SPA with reserve) that attains it.
As a byproduct of our result, we show that the first-stage randomized auction used by the mechanism leads to the fastest convergence to truthful bidding from no-regret learning agents. 

To show that an auctioneer
facing learning bidders using MWU must suffer an $\Omega(T^{3/4})$
revenue loss compared to the setting when it is facing rational agents,
we break down the revenue loss into two non-overlapping epochs: 
one where the learning bidders have not converged to truthful bidding,
and the other where the bidders are truthful. Now an auctioneer using the same auction throughout
the interaction faces a trade-off: they can speed up the learning process to reduce the revenue loss from the first epoch, but this  loses revenue
in the second epoch due to the fact that the auction now differs significantly from SPA.
Our result optimizes this trade-off to show that an 
$\Omega(T^{3/4})$ revenue loss is unavoidable. This naturally suggests 
decomposing the interaction into two epochs: in the first one, the auctioneer
uses a truthful auction to facilitate the convergence to truthful bidding, and in the second one it uses SPA. We then design an auction
that guarantees the fastest convergence to truthful bidding for mean-based
learners in the prior-free setting, and we show that an improved revenue loss
of at most $\widetilde{O}(\sqrt{T})$ can be achieved with this approach. (Importantly, to maintain truthfulness, the decisions of the auctioneers are fixed before the beginning of the interaction and are not affected by the bids.)
This regret of $\widetilde{O}(\sqrt{T})$ seems surprising, because in traditional no-regret
learning settings the optimal regret is achieved when the \emph{exploration}
and \emph{exploitation} phase are intermixed.

% \andresnote{do we need an overview of results?}
% \grigorisnote{\subsection{Overview of Results and Techniques}\label{sec:overview-of-results-and-techinques}

% The key insight
% that the study of deterministic auctions reveals
% is that we need to come up with a different auction format that 
% gives some non-zero utility gain to the low valuation bidder
% when it is bidding truthfully. This is precisely what 
% \emph{strictly IC} auctions can achieve, which require \emph{randomization.}
% Then, in \Cref{sec:randomized-truthful-auctions} we can show that, perhaps surprisingly, there is a way to carefully
% combine \emph{any} deterministic auction $A$ with a randomized one so that,
% asymptotically, the new auction has the same behavior as $A,$ but, crucially,
% conveys information to the low bidder that truthful bidding
% is strictly preferable. 

\subsection{Related Work}\label{sec:related-work}
% \grigorisnote{I have created a further related work
% file for the appendix. feel free to move things there.}
The most closely related works to our setting
are \citet{feng2021convergence, deng2021firstPrice, kolumbus2022auctions, banchio2022artificial,rawat2023designing}. All these works
study the long-term behavior of bidding algorithms that participate in 
repeated auctions, focusing on first-price and second-price auctions,
but they give qualitatively different results. This is because they
make different assumptions across two important axes: the type of
learning algorithms that the bidders use and whether their 
valuation is \emph{persistent} across the interaction
or it is freshly drawn in each round. \citet{feng2021convergence}
studied the convergence of no-regret learning algorithms
that bid repeatedly in second-price and first-price auctions, where
all agents have i.i.d. valuation that are redrawn in every round
from a discrete distribution that has non-negligible mass on each point.
They show that in this setting the bidders exhibit the same-long
term behavior in both second-price and first-price auctions that 
classical theory predicts, i.e., the bids in second-price auctions
are truthful and the bids in first-price auctions form
Bayes-Nash equilibria. 
% \andresnote{i think the lit review should look more on how our paper connects to other. e.g. kolumbus look for SPA not quant. resutls . us extend their resutls and show that if the auction is randomzies it solves the bad properties of spa.}
\citet{kolumbus2022auctions} studied the
same setting with the crucial difference that agents' valuations are persistent across the execution and they
are not resampled from some distribution at every iteration.
Interestingly, they showed that in the case of two bidders with
in second-price
auctions, the agent that has the highest valuation will end up
bidding between the low valuation and its valuation,
whereas the agent with the low type will end up bidding 
strictly below its valuation. Intuitively, in their setting
the high type bidder quickly learns to bid
above the valuation of the low type bidder and always win the auction,
and thus the low type does not get enough signal to push its bid distribution
up to its valuation. On the other hand, when the valuations are redrawn
as in \citet{feng2021convergence}, the competition that the agents face varies.
In the long run, this gives enough information to the algorithms
to realize that bidding truthfully is the optimal strategy. In the 
case of first-price auctions where the agents have
persistent valuations,
%\anupamnote{Similar to what?}\grigorisnote{I rephrased it, does it make more sense now?}
both \citet{kolumbus2022auctions, deng2021firstPrice} provide
convergence guarantees of no-regret learning algorithms.
% \andresnote{this para seems a bit from what we do in this paper. }
The type of ``meta-games'' we touch upon in our work,
where we want to understand the incentives of the agents
who are submitting their valuations to bidding algorithms
that participate in the auctions on the behalf of these
agents,
were originally studied by \citet{kolumbus2022auctions}
and, subsequently, for more general classes
of games by \citet{kolumbus2022and}. 

The pioneering work of \citet{hart2000simple} showed that when players
deploy no-regret algorithms to participate in games
they converge to \emph{coarse-correlated equilibria}.
Recently, there has been a growing interest in the study
of no-regret learning in repeated auctions. The empirical
study of \citet{nekipelov2015econometrics} showed that the bidding
behavior of advertisers on Bing is consistent
with the use of no-regret learning algorithms that bid 
on their behalf. Subsequently, \citet{braverman2018selling}
showed, among other things, that when a seller faces a no-regret buyer 
in repeated auctions and can use non-truthful, it can extract the whole welfare
as its revenue. 
A very recent work \citep{cai2023selling} extended
some of the previous results to the setting with multiple
agents. {For a detailed comparison 
between our work and \citet{cai2023selling}, we refer to \Cref{app:related-work}.}

\citet{banchio2022artificial,rawat2023designing} diverge from the previous works
and consider agents that use $Q$-learning algorithms instead
of no-regret learning algorithms. Their experimental findings show that
in first-price auctions, such algorithmic bidders exhibit
collusive phenomena, whereas  they converge
to truthful bidding in second-price auctions. One of the main reasons for these
phenomena is the \emph{asynchronous} update used by the $Q$-learning
algorithm. The collusive behavior of such algorithms 
has also been exhibited in other settings
\citep{calvano2020protecting,asker2021artificial,asker2022artificial,den2022artificial,epivent2022algorithmic,asker2022impact}.
Notably, \citet{bertrand2023qlearners} formally proved that $Q$-learners
do collude when deployed in repeated prisoner's dilemma games.

In a related line of work, \citet{zhang2023steering}
study the problem of steering no-regret learning agents to 
a particular equilibrium. They show that the auctioneer can use 
\emph{payments} to incentivize the algorithms to converge
to a particular equilibrium that the designer wants them to.
% \grigorisnote{TODO: mention something more about the setting they consider.}
An interpretation of our results is that \emph{randomization}
is a way to achieve some kind of equilibrium steering in repeated auctions.
% {\andresnote{using a truthful mechanisms (which i assume they don't claim it can be done via truthful manner. I think this paper is quite important to mention.}}

Diverging slightly from the setting we consider, some recent papers
have illustrated different advantages of using randomized auctions over deterministic ones. \citet{mehta2022auction, liaw2023efficiency}
showed that there are randomized auctions which induce equilibria
with better welfare guarantees for value-maximizing autobidding agents compared to deterministic ones. In the setting of revenue maximization in the presence
of heterogeneous rational buyers, \citet{guruganesh2022prior} showed that randomization
helps when designing prior-free auctions with strong revenue guarantees, when the valuations
of the buyers are drawn independently from, potentially, non-identical distributions.

\section{Model}\label{sec:setting}

Our model follows the setup used in \citet{kolumbus2022auctions}. There are $T$ rounds, and the auctioneer sells a single item in each round $t=1,\ldots, T$. There are two bidders, with bidder $i \in \{1,2\}$ having a persistent private valuation $v_i$ drawn i.i.d.\ over the discrete set $B_{\Delta} := \left\{0,\nicefrac1\Delta,\nicefrac2\Delta,\ldots,1\right\}$ from a regular distribution $F$. (A discrete distribution is \emph{regular} if the discrete virtual valuation function $\phi(v) := v - \frac 1 {\Delta}\frac{\sum_{v'>v}\Pr[v']}{\Pr[v]}$ is non-decreasing.)
Given an allocation probability $x$ and price $p$, the bidder with valuation $v$ receives a payoff of $v\cdot x - p$.
In what follows, we refer to the bidder with valuation $v_L = \min\{v_1,v_2\}$ (resp.\ $v_H =\max\{v_1,v_2\}$)
as the \emph{low type} (resp.\ \emph{high type}). 

We are interested in truthful auctions,  (also called strategy-proof auctions, or dominant-strategy incentive-compatible mechanisms) that are individually rational, so that at every round $t$ the auctioneer uses a mechanism $((x^t_1,x^t_2),(p^t_1,p^t_2))$ satisfying 
\begin{alignat*}{2}
    v_i\cdot x_i^t(v_i,b') - p_i^t(v_i,b') &\geq v_i \cdot x_i^t(b,b') - p_i^t(b,b'), &\qquad& \forall
    v_i,b,b' \in B_\Delta,\, i=1,2 \,, \\
     v_i\cdot x_i^t(v_i,b') - p_i^t(v_i,b') &\geq 0, &&\forall
    v_i,b' \in B_\Delta,\, i=1,2 \,.
\end{alignat*}
%\andresnote{I don't understand what we want to add with this: We remark that Myerson's lemma \citep{myerson1981optimal}shows that an allocation rule $x(\cdot,\cdot)$ can be coupled with a unique payment rule $p(\cdot,\cdot)$ to satisfy this propertyif and only if for any fixed $b' \in B_\Delta$ it holds that $x(\cdot,b')$ is non-decreasing. }

In this work, we study various properties of
\emph{randomized} truthful auctions. 
\begin{definition}[Randomized Truthful Auction]
A truthful auction $((x_1,x_2),(p_1,p_2))$ is randomized if there is some bid profile $(b_1, b_2) \in B_\Delta$ such that either $x_1(b_1,b_2) \in (0,1)$ or $x_2(b_1,b_2) \in (0,1).$ 
\end{definition} 
%\grigorisnote{I think $\Delta$ should be thought of as the number of different bids and $1/\Delta$ as the differencebetween consecutive bids, so maybe you mean $1/\Delta$ in the above?}\andresnote{correct.fixed.}

%\andresnote{maybe we should add here two claims: (1) the truthful mechanism iff x is monotone and p satisfy the integral equation (not sure if the iff works for a discrete environment though; (2) show that if bidders are bidding truthfully and valuations are idd the SPA is the optimal mechanism among the one that always allocates (again, not sure if it holds on the discrete setting.}
%\grigorisnote{I also wasn't sure how exactly we want tophrase Myerson's lemma. I just wanted to briefly mention that the class of auctions that satisfies the IC, IR requirement is known in our setting and our focus is on the design of the allocation rule since the payment rule is implied by it (up to some discretization error maybe?).} \andresnote{i think this paper might be helpful for discrete type mech. desing: \url{https://www.sciencedirect.com/science/article/abs/pii/S016517650500412X}}.

% \subsection*{Mean Based no-Regret Learning Algorithms}

Bidders employ \emph{learning algorithms} that bid over the $T$ rounds. We assume that the learning algorithms are \emph{mean-based no-regret} learning algorithms \citep{braverman2018selling}. For the following discussion, define $U_i^t(b \mid \mathbf{b}^t_{-i} ) :=  \sum_{\tau=1}^t v_i\cdot x_i^\tau(b,  b^\tau_{-i}) -p_i^\tau(b,  b^\tau_{-i})$ to be the cumulative reward of agent $i$ when they bid $b$ over the $t$ rounds, whereas the other agent's bids are $\mathbf{b}^t_{-i} = \{b^\tau_{-i}\}_{\tau \in [t]}$.
The mean-based property states that if a bid $b \in B_\Delta$ has performed significantly better than bid $b' \in B_\Delta,$ then the probability of bidding $b'$ in the next round is negligible. This is formalized below.

\begin{definition}[Mean-Based Property \citep{braverman2018selling}]\label{def:mean-based-learner}
   % Let $n$ be the number of agents and conside some agent $i \in [n].$ For any $t \in [T]$, let $\tilde b^t_{-i} \in B_\Delta^{n-1}$  be an arbitrary sequence of bids submitted by    the agents $j \in [n]\setminus \{i\}.$ 
    %we denote by $U_i^t(b |\mathbf{b}^t_{-i} ) =  \sum_{\tau=1}^t v_i\cdot x_i^\tau(b, \tilde b^\tau_{-i}) -p_i^\tau(b, \tilde b^\tau_{-i})$ the cumulative reward agent $i$ gets by using a fixed bid $b$ over the $[t]$ rounds.
    %\andresnote{this should be vector of bids for $t$?}\grigorisnote{I changed the notation slightly, does it make more sense now? $b_{-i}$ is vector of bids, if you prefer to use bold for it I'm fine with it. If we stick to this notation, I will also change the no-regret definition.}\andresnote{let's use bold for vectors :)}.
    An algorithm for agent $i$
    is \emph{$\delta$-mean-based} if for any bid
    sequence $\mathbf{b}^t_{-i}$ such that $U_i^{t-1}(b\mid\mathbf{b}^t_{-i}) - U_i^{t-1}(b'\mid\mathbf{b}^t_{-i}) > \delta \cdot T$, 
    for some $b, b' \in B_\Delta$, the probability of playing
    bid $b'$ in the next round is at most $\delta$. We say that an algorithm
    is mean-based if it is $\delta$-mean-based for some $\delta = o(1).$
   % \andresnote{is the def. still correct?}\grigorisnote{Looks good to me!}
\end{definition}

The no-regret learning property states that the cumulative utility that the bidding algorithm generates is close to the cumulative utility that the optimal fixed bid would have generated, regardless of the history of bids the other bidders played. This is formalized in \Cref{def:no-regret}.
% \begin{definition}[No-Regret Learning]\label{def:no-regret}
% Let $\{b^t_i\}_{t \in [T]}$ be the bids submitted by bidder $i \in [n]$,
% $\{\tilde{b}^t_{-i}\}_{t \in [T]}$ be the bids submitted by the rest
% of the bidders,
% and 
% $\{x^t(\cdot,\cdot),p^t(\cdot,\cdot)\}_{t \in [T]}$ be the allocation rule and  payment rule used in the rounds $[T] = \{1,\ldots,T\}$. We say that the algorithm $i$
% \emph{achieves the no-regret property} if 
% \[
%      \E\bigg[\max_b \bigg\{ \frac{1}{T}\bigg(\sum_{t \in [T]}v\cdot x^t(b,\tilde{b}_{-i}^t) - p^t(b,\tilde{b}_{-i}^t)\bigg)\bigg\} - \frac{1}{T}\bigg(\sum_{t \in [T]}v\cdot x^t(b^t,\tilde{b}_{-i}^t) - p^t(b^t,\tilde{b}_{-i}^t)\bigg)\bigg] = o(1) \,
% \].
% %i.e., no fixed bid $b$ can achieve a much higher average utility, for any sequences of opposing bids, and allocation and payment rules,where the expectation is taken with respect to the random draws of the bids $b^1,\ldots,b^T.$ 
% \end{definition}
Mean-based no-regret learning algorithms are becoming a standard class of learning algorithms to use in auction environments (see, e.g., \citet{braverman2018selling, feng2021convergence, deng2021firstPrice, kolumbus2022auctions}, and references therein) and include many known no-regret learning algorithms, including
the multiplicative-weights update algorithm (MWU). %(In \Cref{??} we provide a more detailed overview of such algorithms.) 
For completeness, we present the version of MWU that
we use in our work in \Cref{alg:mw}.
The above definitions consider 
a fixed value of $T.$ Thus, given 
a sequence of such values $T$ and the limiting
behavior as $T \rightarrow \infty$,
we say that a family of algorithms, parameterized
by the time horizon $T$, satisfies the mean-based 
definition if there exists $\{\delta_T\}_{T \in \mathbbm{N}}$ such that $\delta_T \rightarrow_{T \rightarrow \infty} 0,$ and each
algorithm in this family is $\delta_T$-mean-based.
We define the no-regret property of such
a family of algorithms in a similar way.
In general, the asymptotic behavior
of the algorithms we study in this work 
is with respect to $T$ and the big $O$ notation 
suppresses quantities that do not depend on $T.$

%In terms of feedback information for the algorithms, 
For the sake of exposition, we focus on the \emph{full feedback} setting:
after every round $t \in [T]$, the algorithm learns for each bid $b \in B_\Delta$ the (expected) utility it would have generated had it played bid $b$. In \Cref{app:extensions}, we discuss
potential extensions.
%\andresnote{we need to add a discussion section at the end of the paper.}

Throughout this paper we make a natural assumption on the algorithms which restrict bidders to never bid over their value. Specifically, 
%\begin{assumption}
for any round $t$, and any history of bids before period $t$, agent $i$ bids $b_i>v_i$ with zero probability. 
%\end{assumption}
Without this assumption, \citet{braverman2018selling,cai2023selling} show that the auctioneer can extract the entire welfare in the setting where the valuations of the agents are drawn i.i.d.\ in each round. 
We focus on the \emph{last-iterate} convergence
of the distribution of the bids of the algorithms
as $T \rightarrow \infty.$ This is a 
% Last iterate converge is a very 
desirable
property of algorithms in multi-agent
games, and recent work has focused
on establishing it for learning algorithms
\citep{cai2022finite,cai2022accelerated1,cai2022accelerated}. %\anupamnote{Do we cite other people as well? Not clear it is required.} 
This is formalized in \Cref{def:last-iterate-convergence}.

Due to space limitations, all the proofs of our results
can be found in the appendix.

\section{Deterministic Truthful Auctions}\label{sec:deterministic auctions}
In this section we study the effect of the 
learning rate on the convergence of no-regret
learning 
algorithms in {non-degenerate} \emph{deterministic} truthful auctions. 
{Informally, the non-degeneracy 
requirement states that
\textbf{i)} the winning bidder $W$ under truthful bidding
gets strictly positive utility, \textbf{ii)} there is some sufficiently small bid of the winning
bidder such that the runner-up bidder $R$ wins
the item by bidding $v_R$ but does not win
by bidding $v_R - \nicefrac{1}{\Delta}$. The formal definition is given in \Cref{def:non-degenerate}.}
We focus our attention to bidders that use MWU
to participate in the auctions
and we study the bidding distribution they converge to
as a function of the ratio of the learning rate of their algorithms.
Throughout this section we refer to the bidder
who wins the auction under truthful bidding as 
the winning bidder and to the bidder that loses
the auction under truthtelling as the runner-up
bidder. Our main result in this section
shows the following behavior in {non-degenerate} deterministic
truthful auctions:
\begin{itemize}
    \item The winning bidder converges to bidding
     between its minimum winning bid and its
    true value, no matter what the choice of the learning 
    rates of the algorithms are.
    \item If the learning rate of the runner-up bidder
    is strictly faster 
    than the learning rate of the winning bidder,
    then the runner-up bidder converges to bidding
    truthfully.
    \item If the learning rate of the runner-up bidder
    is not strictly faster than that
    of the winning bidder, then the runner-up
    bidder converges to a bidding distribution
    whose mean is strictly smaller than its true value. {This result holds
    under an even milder requirement 
    than non-degeneracy. Namely,
    as long as the utility of the winning bidder
    under truthful bidding is strictly positive.
    }
\end{itemize}

We remark that, when the learning rates of the algorithms
are instantiated before the random draw of the two valuations of the agents that are i.i.d. from some distribution $F$,
then with probability at least $1/2$ the runner-up
bidder will not converge to bidding truthfully, if the
underlying auction is deterministic. As we will show
later, this behavior worsens the revenue guarantees of the auction. 
Let us first set up some notation to facilitate
our discussion. We denote by $v_W \in \{v_L, v_H\}$ and $\eta_T^W$ (resp., $v_R \in \{v_L, v_H\}$, and $\eta_T^R$) the value and learning rate 
of the winning bidder (i.e., the one who wins if both bidders bid truthfully) and the runner-up bidder, respectively. 
% The restriction we place on $\eta_T^W$ and $\eta_T^R$ is
% that they are non-degenerate, in the sense that they
% allow the bidders to achieve the no-regret guarantee.
We would like to remind the readers that, typically,
the learning rate $\eta_T$ of MWU is a decreasing
function of $T$ and is
chosen in 
a way to minimize the quantity
$
    \nicefrac{C_\Delta}{\eta_T} + C'_\Delta\cdot \eta_T \cdot T\,,
$
where $C_\Delta, C'_\Delta$ are discretization-dependent constants. Usually, it is instantiated with $\eta_T = 1/\sqrt{T}.$
However, for the purposes of our analysis we will say
that $\eta_T$ is \emph{non-degenerate} if
% \[
%     \frac{C'_\Delta}{T} \leq \eta_T \leq \frac{C''_\Delta}{\log T} \,,
% \]
$
    \lim_{T \rightarrow \infty} \eta_T \cdot T = \infty, \lim_{T \rightarrow \infty} \eta_T \cdot \log T = 0 \,.
$
 The intuition is that if the learning rate
is slower than $1/T,$ the bidder will be adjusting its
bid distribution very slowly, so it will not learn to 
bid correctly. On the other hand, if the rate is faster
than $1/\log T$ then the bidder will be adjusting 
its distribution too aggressively. 
% (This bound of 
%  $1/\log T$ arises for a technical reason: it helps
%  us prove a union bound in our proof.)
 
Our results show that in \emph{deterministic} auctions
the convergence behavior of the bidders depends
heavily on the ratio between the learning rates.
In particular, for the bidder with valuation $v_W$, we show that its bids converge to a distribution supported between $\hat p$, the price it would pay if both bidders bid truthfully, and its value $v_W$, no matter
what the choice of the learning rate of its algorithm is.
%\andresnote{should $\hat p = v_L$?}\grigorisnote{I think not necessarily, $\hat p$for example could be the reserve price}.\andresnote{ok.} 
On the other hand, the convergence behavior of
the runner-up bidder is more nuanced:
if $\nicefrac{\eta_T^R}{\eta_T^W} = \omega(1),$ i.e.,
the runner-up bidder learns more aggressively than
the winning bidder, then it converges to 
 bidding truthfully.
%  \footnote{For this case, to make the exposition cleaner, we state the formal result only for second price auctions. However, it holds for more general deterministic auctions.}
 However,
 if $\nicefrac{\eta_T^R}{\eta_T^W} < C_\Delta,$ where $C_\Delta$ is some discretization-dependent
 constant, then the runner-up converges to a bidding
 distribution that puts positive mass on every (discretized)
 point between $0$ and $v_R,$ and, in particular, its 
 expected value is strictly less than $v_R.$
We remark that even though our proof idea is inspired by  \citet{kolumbus2022auctions}, 
our analysis considers all the possible learning
rates that MWU could be instantiated with and requires
a more technically involved argument.
In particular, we notice that while the result of \citet{kolumbus2022auctions} is, implicitly, proved for \emph{identical} learning rates, we show that the choice of the learning rate affects the qualitative behavior of the algorithms in a crucial way.
 
 %The notion of convergence we are studying here
 %is the one defined in \Cref{??}.
 %This is formalized in the following claim.
%Even though our proof that follows is inspired by the
%ideas of \citet{kolumbus2022auctions}, there are some details
%that make it more involved technically. 
%In particular, the result of \citet{kolumbus2022auctions} is, implicitly, proved for \emph{constant} learning rate. As we will see, the choice of the learning rate affects the qualitative behavior of the algorithm in a crucial way.
We prove this result in two parts. We start with the case
where $\nicefrac{\eta_T^R}{\eta_T^W} < C_\Delta.$ The idea of 
the proof is to split the horizon into 
consecutive periods of size $O(1/\eta_T^R),$ which we call
\emph{epochs}. Now following the idea of \citet{kolumbus2022auctions}, we show that 
within each epoch the runner-up bidder bids
truthfully $\Omega(1/\eta_T^W)$ many times, so the
total
utility of the winning
bidder for bidding between $\hat p$ and $v_W$
will be at least $\Omega(1/\eta_T^W)$ greater than
bidding anything between $0$ and $\hat p - 1/\Delta.$
Because its learning rate is $\eta_T^W,$ this means
that it will move a constant fraction of its mass from
the region $\{0,1/\Delta,\ldots,\hat p - 1/\Delta\}$
to the region $\{\hat p, \ldots, v_W\}.$ Summing
this geometric series, we see that the winning bidder
will submit bids in the region $\{0,1/\Delta,\ldots,\hat p - 1/\Delta\}$ at most $O(1/\eta_T^W)$ many times.
Let us now focus on the runner-up bidder. Following
the previous argument, its total utility for bidding
$v_R$ will be at most $O(1/\eta_T^W)$ greater than bidding
some other bid $b' \in B_\Delta.$ Since $\eta_W^R/\eta_T^W < C,$ this means  the probability of bidding $b'$ after
$T$ rounds is only smaller than the probability of 
bidding $v_R$ by a discretization-dependent multiplicative
constant. The formal statement of this result and its 
proof follow are postponed to \Cref{app:deterministic-auctions}.

\begin{comment}
\begin{theorem}[No Deterministic Auction Leads to Truthful Bidding]\label{theorem:deterministic anonymous auctions do not lead to convergence}
Given bidder valuations $v_L < v_H$, suppose bidders use MWU with rate $\eta = O(1/\sqrt{T})$ to bid on the auction.
\andres{[Can we say this in a footnote:]where we are suppressing constants that depend on the discretization
parameter $\Delta$.}
Assume w.l.o.g. that $x_H(v_L,v_H) = 1, x_H(v_L, v_H) = 0, p_H(v_L, v_H) = \hat{p} < v_H.$
Then, with probability at least $1 - C_\Delta^1\cdot\sqrt{T}\cdot e^{-C_\Delta^2\sqrt{T}}$,
where $C_\Delta^1, C_\Delta^2 > 0,$ are some constants that depend on $\Delta$,
the high type will converge to bidding
between $\hat p, v_H$ and the low type
will converge to a bidding distribution with
$0 < \Pr[0] \leq \Pr[1/\Delta] \leq \ldots \leq \Pr[v_L].$
\end{theorem}
\end{comment}

% \Cref{theorem:deterministic anonymous auctions do not lead to convergence} showed when the learning
% rate of the winning bidder is at least as 
% fast as the rate of the runner-up bidder, up to
% some discretization dependent constant, deterministic auctions do not lead to truthful bidding of the runner-up
% bidder. 
Our next result illustrates that the convergence
behavior of the runner-up type exhibits
a sharp phase-transition phenomenon: if $\eta_T^R$ is
even slightly faster than $\eta_T^W,$ i.e.,
 $\nicefrac{\eta_T^R}{\eta_T^W} = \omega(1),$
then the runner-up will learn to bid truthfully.
Let us first give a high-level idea of the proof.
Similarly as before, we split the horizon
into intervals of size $O(1/\eta_T^W).$ We consider
the first interval of this interaction. Because 
of the choice of the learning rate, we can show
that the winning bidder will bid $v_R - 1/\Delta$
at least $\Omega(1/\eta_T^W)$ many times. Thus,
this means that the total utility of bidding
$v_R$ for the runner-up bidder will be at least
$\Omega(1/\eta_T^W)$ greater than bidding
any other bid. Since $\eta_T^R/\eta_T^W = \omega(1)$,
after the first epoch the MWU algorithm
will place all but a $o(1)$-fraction of the probability
mass to bidding truthfully. The formal statement  and its proof appear in \Cref{app:deterministic-auctions}.

Next, we discuss the implications that our 
results have to the revenue guarantees of the auctioneer.
In the setting with rational bidders, the seminal work of \citet{myerson1981optimal} showed
that using second-price auctions with an anonymous reserve price, which
depends on the value distribution $F$, generates the optimal revenue for the auctioneer. Our next result shows that this is no longer true
when the bidders are learning agents, even when the valuations of the 
agents are
drawn i.i.d. from the uniform distribution on $B_\Delta$, which we
denote by $U[B_\Delta].$
Intuitively, this happens because, no matter what the reserve price
is,
% \footnote{With the exception of setting the reserve price to be $1-1/\Delta,$ where the only revenue contribution will be happening whenever the valuation of one of the agents is exactly 1.}
with some non-zero probability the valuations of both agents
will be higher than the reserve price. Then, since the runner-up
bids will be strictly lower than the true valuation, the generated
revenue will be strictly lower than in the setting with rational 
agents, even when $T \rightarrow \infty.$ 
% This is formalized
% in the following result.
% \grigorisnote{To get the result for general
% truthful deterministic auctions we need to claim that in the 
% learning environment, the revenue optimal auction has the property
% that when the low type bids below its valuation, then the revenue does not
% increase, and for some combinations of $v_L,v_H$, when $v_L$ bids way
% below its valuation then the revenue strictly decreases. I am pretty sure this is correct, but I am not sure how to prove it formally. There
% are some annoying auctions in which underbidding increases
% the revenue, so we need to argue that all these pathological cases
% will not be revenue optimal in the learning setting.}

% \andresnote{IIUC, this is saying that even with reserve price determinsitc symmetric aucitons don't hhelp. Shouldn't this result be in the previous section??}
% \grigorisnote{You are right, this makes more sense.
% I moved it.}
\begin{theorem}[SPA with Reserve Is Not Revenue Optimal]\label{thm:deterministic-auctions-not-revenue-optimal}
    Let two agents draw their valuations from the uniform distribution
    over $U[B_\Delta]$ and participate in $T$ repeated auctions using
    mean-based learners. Let $b_1^T, b_2^T$ be the bid distributions
    after $T$ rounds.
    Let $\mathrm{Rev}(b_1,b_2;r)$ denote the revenue of the second-price
    auction with reserve price $r$ when the bids are $b_1, b_2 \in B^2_\Delta.$ Then, for all $r < 1-1/\Delta,$
    \[
         \E_{v_1,v_2 \sim U[B_\Delta]}\left[\lim_{T \rightarrow \infty}\E_{b_1 \sim b_1^T, b_2 \sim b_2^T}[\mathrm{Rev}(b_1,b_2;r) \mid v_1,v_2] \right] < \E_{v_1,v_2 \sim U[B_\Delta]}\left[\mathrm{Rev}(v_1,v_2;r)\right] - c \,,
    \]
    where $c > 0$ is some constant that does not depend on $T.$
\end{theorem}

\section{The Value of Randomized Truthful Auctions: The Asymptotic Case}\label{sec:randomized-truthful-auctions}

In this section we show that there is a class of randomized
auctions such that when mean-based no-regret learners participate 
in them repeatedly, they converge to \emph{truthful} bidding.
This holds for any choice of the learning
rates of these algorithms, which is in contrast 
to the results of \Cref{sec:deterministic auctions}.
We start by defining a class of auctions called
\emph{strictly} IC.

\begin{definition}[Strictly IC Auctions]\label{def:strictly-IC-auctions}
    An auction is called strictly IC if for every bidder $i \in [n]$,
    valuation $v_i \in B_\Delta$, and bid profile $b_{-i} \in B_\Delta^{n-1}$
    it holds that $
        v_i \cdot x_i(v_i, b_{-i}) - p_i(v_i,b_{-i}) > v_i \cdot x_i(b, b_{-i}) - p_i(b,b_{-i}), \forall b \neq v_i \,.
    $
\end{definition}

The next result, which is very useful for our derivation, states
that when mean-based no-regret learning algorithms bid in
some strictly IC auction they converge to bidding truthfully. Recall the definition of a 
mean-based learner (cf. \Cref{def:mean-based-learner})
which states that if the cumulative utility of
some bid $b$ up until round $t-1$
is much smaller than the utility
of some other bid $b'$, then the probability
of playing $b$ in the next round $t$ is negligible.
The proof is postponed to \Cref{app:randomized-asymptotic}.
% \grigorisnote{I think we can have two results. One for general
% no-regret learners where we have ``on average'' convergence to 
% truthful bidding and one for mean-based no-regret learners
% where we have last-iterate convergence to truthful bidding.}

% \grigorisnote{Can we get last iterate convergence for general 
% no-regret learners?}
\begin{lemma}[Convergence in Strictly IC Auctions]\label{lem:convergence-strictly-IC-auctions}
Consider $n$ bidders that participate in a repeated
strictly IC auction $A$ using
mean-based no-regret learning algorithms. Then, as $T \rightarrow \infty,$ the bidders converge to truthful bidding in a last-iterate sense.
\end{lemma}

The next important observation is that 
when we are taking a non-trivial
combination of an IC auction
with a strictly IC auction, the resulting
auction is strictly IC. The notion
of mixture we consider is formalized in
\Cref{def:mixture-auctions}.

\begin{definition}[Mixture of Auctions]\label{def:mixture-auctions}
Let $A = \left(x(\cdot), p(\cdot)\right)$ be an IC auction and $A' = \left(x'(\cdot), p'(\cdot)\right)$
be a strictly IC
auction. For some $q \in (0,1)$ 
we define the $q$-mixture of the auctions 
$A, A'$ to the be auction $\widetilde{A}_q = \left(q \cdot x(\cdot) + (1-q)\cdot x'(\cdot), q \cdot p(\cdot) + (1-q)\cdot p'(\cdot)\right).$
\end{definition}

Notice that for the allocation rule 
$q \cdot x(\cdot) + (1-q)\cdot x'(\cdot)$
Myerson's lemma states that the corresponding payment rule that makes the auction truthful is
indeed $q \cdot p(\cdot) + (1-q)\cdot p'(\cdot).$
The following claim, whose proof follows from the definition of 
this class of auctions, formalizes the fact that the class of strictly IC auctions is closed under mixtures with IC auctions.

\begin{claim}[Mixture of IC and Strictly IC Auction]\label{clm:mixture-IC-strictly-IC}
    Let $A, A'$ be an IC, strictly IC auction, respectively.
    Then, for any $q \in (0,1)$ the auction $q\cdot A + (1-q) \cdot A'$ is strictly IC.
\end{claim}
% The proof follows immediately from the definition.
{We remark that we can construct strictly IC
auctions using randomization; such an example
is presented in \Cref{sec:no-regret-learning-auctioneer}.} Equipped with the above results, we can show that there 
is a black-box transformation from any IC auction $A$ to 
a strictly IC auction $A'$ so that as $T\rightarrow \infty,$ any mean-based
    learning algorithms
    converges to truthful bidding, and the auction $A'$ is close to the auction $A$
    in the sense that $|x_i(b) - x'_i(b)| = o(1), |p_i(b) - p'_i(b)| = o(1), \forall i \in [n], \forall b \in B^n_\Delta.$
    %andresnote{is this uniformly on the bids?}\grigorisnote{Yes, I just edited it. Does it make sense now?}\andresnote{i think so. i'm not a big fun of o notation, but i think is correct. btw,shouldn't that be in $B_\Delta^2$?}\grigorisnote{Yes, I fixed it. Thanks!}
The formal statement of the result follows.
\begin{theorem}\label{thm:mixture-of-IC-with-strictly-IC-last-iterate-convergence}
Let $A$ be an IC auction for $n$ agents with
valuations $v_1, \ldots, v_n$. Let each agent $i \in [n]$
use a mean-based no-regret learning algorithm to bid in the auction.
Then, there exists an auction $A'$ such that 
for each agent $i \in [n]$
we have that $\lim_{T \rightarrow \infty} b^T_i =  v_i$ and $|x_i(b) - x'_i(b)| = o(1), |p_i(b) - p'_i(b)| = o(1),\forall b\in B^n_\Delta,$ where $x_i(\cdot), x'_i(\cdot)$ (resp. $p_i(\cdot), p'_i(\cdot)$) is the allocation (resp. payment) rule of $A, A'.$
\end{theorem}
%is postponed to \Cref{app:randomized-asymptotic}.

% \grigorisnote{TODO: add examples of strictly IC auctions
% and explain the intuition. We can have some linear auction, proportional
% allocation, soft-max etc.}
% \andresnote{yes, this would be great}

%\medskip
\noindent\textbf{Equilibria of Meta-Game in Repeated Strictly IC Auctions}\label{sec:meta-game-strictly-IC}
We now describe the implications that our results have
for the meta-game that we alluded to in \Cref{sec:intro}. Recall that this game is defined as follows: the agents submit their valuations
to mean-based no-regret learning algorithms and then, given
these fixed valuations,
they bid on the behalf of the agents in a repeated truthful auction $A$.
The main question we are interested in understanding is
given the specification of the auctions and the valuations
of the agents, what is the optimal value they should
submit to the algorithms in order to maximize their utility,
after a large number of steps?

Despite the fact that $A$ is IC and IR, \citet{kolumbus2022auctions}
showed that, rather surprisingly,
when two agents use MWU to participate
in repeated second price auctions there are
instances where the agent
with the low valuation has an incentive to report
a higher value to its algorithm than its true one.
This is because the valuation reported by one
agent affects the bidding distribution that
the other agent will converge to. To illustrate
this point, assume that the low type reports
$v'_L > v_H$ to its bidding algorithm. Then, the 
bidder with type $v_H$ will take the role of the low
bidder in the interaction and will 
converge to bidding strictly below
$v_H.$ Now if its expected bid is also below $v_L$, this
will generate strictly positive utility for its opponent.
Using our previous construction from  \Cref{thm:mixture-of-IC-with-strictly-IC-last-iterate-convergence}
and transforming the auction $A$ to a strictly
IC auction $A',$
we can show that in the new meta-game every agent
can gain at most $o(1)$ more utility in the long run
by misreporting to the algorithm than reporting its true
valuation. The reason why we observe a
qualitatively different behavior in our construction
is that every algorithm converges to bidding
its reported value, no matter what the reported
values of the other agents are.
Due to space constraints, we refer the interested 
reader to \Cref{app:randomized-asymptotic}

%\medskip
\noindent\textbf{Revenue Maximization in the Learning Setting}\label{sec:revenue-optimal-auctions-repeated-learners}
In this section, we illustrate
another application of \Cref{thm:mixture-of-IC-with-strictly-IC-last-iterate-convergence}
to revenue maximization in the learning setting. 
We are interested in auctions with strong revenue guarantees
when the bids are coming from the limiting distribution
of the algorithms, as $T \rightarrow \infty.$
% Recall
% the revenue maximization problem we are interested in,
% which is a modification of
% the classical revenue maximization problem to account for the fact
% that the bids come from no-regret learning algorithms.
% \begin{itemize}
%     \item There are $n$ agents who draw their bids from i.i.d. regular distributions  $F_1,\ldots,F_n.$
%     \item The auctioneer, who knows the distributions
%     $F_1,\ldots,F_n$, commits to an auction $A.$
%     \item The agents submit their valuations
%     to mean-based no-regret learning algorithms.
%     \item The algorithms bid in $A$ repeatedly for
%     $T$ rounds.
%     \item The goal of the auctioneer is to maximize
%     the revenue it collects
%     under the limiting distribution of the bidders.
% \end{itemize}
% \andresnote{these bullet points are in the model section}\grigorisnote{Do you think it would be worth
% it to remind the reader the setting here?}\andresnote{no, especially that we will be space-consttrained.}
This has the additional
complication that not only do agents draw their valuations
from the distribution $F,$ but
also their bids come from the limiting distribution
that the algorithms converge to, as $T \rightarrow \infty.$
As we have seen already,
this distribution depends on the valuation reported to
the algorithm,
the particular mean-based algorithm that it is using, and,
potentially, the reported valuations and the algorithms 
of the opposing bidders.

As we explained in \Cref{sec:deterministic auctions},
second price auctions with reserves have strictly
worse revenue guarantees in the setting with learning
bidders compared to the setting with rational bidders.
Using our transformation described in \Cref{thm:mixture-of-IC-with-strictly-IC-last-iterate-convergence}
we can restore their revenue guarantees. 
The following result whose formal proof is deferred to \Cref{app:randomized-asymptotic} is,
essentially, a corollary of \Cref{thm:mixture-of-IC-with-strictly-IC-last-iterate-convergence}.
Let us denote by $\mathrm{Rev}(A;b_1,\ldots,b_n)$ the revenue of some auction $A$  and
by $\mathrm{Rev}(\mathrm{Myerson};b_1,\ldots,b_n)$ the revenue of Myerson's optimal auction for $F,$ where the bid profile is $b_1,\ldots,b_n \in B^n_\Delta$.

\begin{corollary}\label{cor:revenue-guarantees-randomized}
    Consider an environment with $n$ agents
    that draw their values i.i.d. from some regular distribution
    $F$ and participate in repeated single-item auctions using 
    mean-based no-regret learning algorithms. Then, there is a randomized auction
    $A$ so that 
    \begin{align*}
        \E_{v_1,\ldots,v_n \sim F^n}\bigg[ \lim_{T \rightarrow \infty}\E_{b_1\sim b_1^T,\ldots,b_n \sim b^T_n}[\mathrm{Rev}(A;b_1,\ldots,b_n)]~\bigg|~ v_1,\ldots, v_n\bigg] \\
        \geq \E_{v_1,\dots,v_n \sim F^n}[\mathrm{Rev}(\mathrm{Myerson;v_1,\ldots,v_n)}] - o(1).        
    \end{align*}
\end{corollary}

Given the results
from \Cref{thm:deterministic-auctions-not-revenue-optimal} and 
\Cref{cor:revenue-guarantees-randomized} we would like to remark the following.
\begin{remark}[Randomized Auctions vs. SPA with Reserve]\label{rem:randomized-vs-spa}
    Our results illustrate that randomized auctions have strictly
    better revenue guarantees compared to SPA with reserve price,
    when the bidders are using mean-based 
    no-regret learning algorithms.
    This is a property of randomized auctions that is not witnessed in the 
    setting where the bidders are fully rational, as proven by \citet{myerson1981optimal}.
\end{remark}

\section{Revenue Maximization in the Finite Time Horizon Setting}\label{sec:no-regret-learning-auctioneer}
So far, we have focused on the asymptotic
regime and we have studied the convergence
of the learning bidders under various auctions. 
In this section, we study the \emph{finite-horizon} setting,
where our goal is to come up with auctions that have
strong \emph{revenue} guarantees for the auctioneer.
% \subsection{Prior-Free Setting}\label{sec:}
We focus on the \emph{prior-free} setting, meaning
that the auctioneer does not have 
any distributional knowledge about the valuation
of the agents.
% \footnote{The distributions could
% be different for each agent.}
Similarly to the rest of the paper,
we assume that the two buyers are
using mean-based no-regret learning algorithms to
participate in single-item auctions for 
$T$ rounds.
Since we are working on the prior-free setting, 
it is natural to compete with the
 cumulative revenue of the second-price auction.
%  \footnote{The results we discuss in this section also hold when the benchmark
%  is second-price auctions with some reserve.}
% For any time horizon $T$, 
% we consider the following natural
% class of auctions
% \[
%     \calA_T = \{p_T \cdot A + (1-p_T) \cdot A': p_T \in [0,1] \} \,,
% \]
% where $A, A'$ are any truthful auctions that are independent of $T.$\grigorisnote{TODO:define the previous set more formally.}
% Let $\mathrm{SPA}$ denote
% the second price auction.
The goal of the auctioneer is to choose
an auction in a way that minimizes 
% the following metric 
\[
    \wt{\reg}_T(A;v_L,v_H) = \sum_{t=1}^T \rev(v_L,v_H;\SP) - \E\left[\sum_{t=1}^T \rev(b_L^t,b_H^t;A) \right] \,,
\]
where the expectation is taken with respect to 
the randomness of the learning algorithms
and, potentially, the auction. We will refer to this benchmark as the
{\emph{auctioneer regret}.}
% Moreover, we are interested in auction
% schedules that have small \emph{variance}, which is 
% defined as
% \[
%     \wt{\vars}(A_1,\ldots,A_T) = 1+\sum_{t=2}^T \mathbbm{1}[A_t \neq A_{t-1}] \,.
% \]
One quantity that will be useful for the derivation
of our results is the following
\[
    \gamma_A = \min_{i\in \{1,2\},b_i, b_{-i}, v_i \in B^3_\Delta: b_i \neq v_i} \left\{ \left( v_i\cdot x_i(v_i,b_{-i}) - p_i(v_i,b_{-i})\right) -  \left( v_i\cdot x_i(b_i,b_{-i}) - p_i(b_i,b_{-i})\right)\right\} \,,
\]
i.e., the minimum increase in the utility by bidding truthfully
instead of bidding non-truthfully in $A.$
% In this section, we assume that the mean-based algorithms
% are parametrized by the same $\eta_T = o(1).$\footnote{If this is not the case, all our results hold by setting $\eta_T = \min\{\eta^1_T,\eta^2_T\}.$} 

% \subsection{Auctioneer Regret Bounds in Terms of $T$}\label{sec:modified-regret-bounds-T}

Our first goal is to understand the dependence of the auctioneer
regret on the time horizon $T$. Then, we will move on to establishing bounds with respect to the number of discretized bids $\Delta.$
Our first result shows that given any strictly IC auction $A$ there exists an auction $A_T$ that achieves auctioneer regret $O\left(T\cdot\sqrt{\frac{\Delta \cdot \delta_T }{\gamma_A}}\right).$ This is formalized below and the proof is postponed to \Cref{app:finite-horizon}.

\begin{proposition}\label{prop:regret-single-auction-upper-bound}
There exists auction $A_T$ which is a mixture of some strictly IC auction
$A$ and $\mathrm{SPA}$ such that, for all $v_L, v_H \in [0,1]^2$ and for all $\delta_T$-mean-based learning
algorithms it holds that
$
     \wt{\reg}_T(A_T;v_L,v_H) = {O}\left(\sqrt{\frac{\Delta \cdot \delta_T}{\gamma_A}} \cdot T \right), \forall v_L, v_H \in B^2_\Delta \,.
$
\end{proposition}

% \begin{remark}\label{rem:T-3/4-regret-upper-bound}
    We emphasize that for common mean-based no-regret learning algorithms
    such as MWU it is the case that $\delta_T = \widetilde O\left({\nicefrac{1}{\sqrt T}}\right),$ which implies 
    that the auctioneer regret from \Cref{prop:regret-single-auction-upper-bound}
    grows as $\widetilde O\left(T^{3/4}\right).$
% \end{remark}
% Given the previous result, we ask if some auction $A_T$  has better guarantees in terms of its auctioneer
% regret. 
Our next result complements this result by showing that even if the high-valuation bidder always bids truthfully and the low-valuation bidder uses MWU with learning rate  $\Theta(\nicefrac{1}{\sqrt{T}}),$
no auction can achieve a better auctioneer regret than $O(T^{3/4}).$
\begin{proposition}[Lower Bound for Constant Auction Policies]\label{prop:lower-bound-constant-auctions}
Consider a repeated auction environment where the high-valuation
bidder bids truthfully and the low-valuation bidder
uses MWU with rate $\Theta(\nicefrac{1}{\sqrt{T}}).$
Then, every truthful auction $A_T$ has an auctioneer regret $\wt{\reg}_T(A_T;v_L,v_H) \geq C_\Delta \cdot T^{3/4},$ where $C_\Delta > 0$ is some constant that
depends on the discretization parameter.
\end{proposition}

The proof is postponed to \Cref{app:finite-horizon}. We note that choosing the learning rate
of MWU to be $1/\sqrt{T}$ gives the optimal no-regret
guarantees. Other choices, such as $\eta_T = \Omega(1),$ 
have trivial regret bounds.
% \anupamnote{The reader may wonder why we don't set the learning rate $\eta_T = \Omega(1)$, in which case the above calculations will only give $\Omega(\sqrt{T})$. Do we say something about that?}

Having established the previous results for repeated auctions
where the auctions remain \emph{constant} across all the iterations,
it is natural to ask whether we can get improved
results when the auctioneer is allowed to change the underlying
auction, but in a way that is \emph{oblivious} to the bids that 
bidders have submitted so far. In other words, the auctioneer
has to commit to an \emph{auction schedule} $\{A_1,\ldots,A_T\}$
\emph{before} the beginning of the interaction. We extend
the definition of the auctioneer regret in a natural 
way to allow for different auctions in every timestep and 
we denote
$
    \wt{\reg}_T(A_1,\ldots,A_T;v_L,v_H) = \sum_{t=1}^T \rev(v_L,v_H;\SP) - \E[\sum_{t=1}^T \rev(b_L^t,b_H^t;A_t)] \,.
$
Our next result shows that there exists an auction schedule where
the auctioneer changes the underlying auction only \emph{once}
throughout the interaction so that its regret is bounded
by $\widetilde O(\delta_T \cdot T).$ For typical choices
of $\eta_T$ this translates to an auctioneer regret bounded
by $\widetilde O({\sqrt{T}}).$
The main insight is that the auctioneer can split the interaction
into two intervals: the first interval has size $T_0,$ for some 
appropriately chosen $T_0 \in [T],$ where the auctioneer uses some strictly IC auction $A$ that encourages 
the learners to converge to bidding truthfully. Then, assuming
that $T_0$ is large enough to guarantee this convergence,
the auctioneer switches to using second-price auction. 
This is perhaps counterintuitive because in other no-regret learning settings,
such as multi-armed bandits, the optimal regret bound is achieved
when \emph{exploration} and \emph{exploitation} are happening
simultaneously, whereas in our setting these two phases
are separated.
% \grigorisnote{does the previous sentence make sense?}\andresnote{i think so but maybe a  bit detour of our point. if you think is worth tom mention i'd suggest to succintly add this to the intro}

\begin{theorem}\label{thm:auction-schedule-regret-upper-bound}
There exists an auction schedule $(A_1,\ldots,A_T)$
in which $A_1 = A_2 = \ldots = A_{T_0} = A,$ where $A$ is any strictly IC
auction,
and $A_{T_0 + 1} = A_{T_0 + 2} = \ldots = A_T = \SP$, 
that 
achieves
$
     \wt{\reg}(A_1,\ldots,A_T;v_L,v_H) \leq O\left(  \delta_T \cdot T \cdot \left(\frac{1}{\gamma_A} + \Delta \right)\right), \forall v_L, v_H \in B^2_\Delta \,.
$
\end{theorem}
The formal proof of this result is postponed to \Cref{app:finite-horizon}.
% \begin{remark}\label{rem:auctio-schedule-sqrt-T-upper-bound}
    The previous result shows that for $\eta_T = \smash{\widetilde O\left({\nicefrac{1}{\sqrt T}}\right)}$ the auctioneer
    regret of the auction schedule we designed is $\smash{\widetilde O(\sqrt{T})}.$
    Thus, we see an $\widetilde O(T^{1/4})$ improvement compared to the
    previous setting where the auctioneer was restricted to 
    be using the same auction across all iterations.
% \end{remark}

Next, we prove that even if the auctioneer uses 
a different auction in every step, our bound from \Cref{thm:auction-schedule-regret-upper-bound} is (almost) optimal with respect to the time horizon $T.$
The proof idea is that when the agents are using MWU with learning
rate $\eta_T,$ the signals in the first $O(1/\eta_T)$
steps are insufficient for them to move their bidding distribution to truthful bids. I.e., with at least
some constant probability in every round within the first $O(1/\eta_T)$
rounds, they will not be bidding their true valuation. 
Importantly, our lower bound holds even in the (unrealistic) setting where the 
auctioneer can choose $A_1,\ldots,A_T,$ conditioned on $v_L, v_H.$
This is formalized
below; the proof is postponed to \Cref{app:finite-horizon}.

\begin{proposition}\label{prop:root-T-lower-bound-all-environments}
When two agents are using MWU with learning rate $\nicefrac{1}{\sqrt{T}}$
to participate in repeated single-item auctions for all the 
auction schedules $(A_1,\ldots,A_T)$ it holds that 
$
    \smash{\wt{\reg}(A_1,\ldots,A_T;v_L,v_H)} = \Omega(\sqrt{T}) \,.
$
% even when the auctioneer can choose $(A_1,\ldots,A_T)$ conditioned
% on $v_L,v_H.$
\end{proposition}

% \subsection{Auctioneer Regret Bounds in Terms of $T,\Delta$}\label{sec:modified-regret-bounds-T-Delta}
Having established the optimal dependence with respect
to the time horizon $T,$ we now shift our attention
to understanding the dependence of the auctioneer regret on the 
discretization parameter $\Delta.$ 
First, we define an auction $\bar A$ that satisfies 
$\gamma_{\bar A} = \Theta(\nicefrac{1}{\Delta^2}).$ 
\begin{definition}[Staircase Auction]\label{def:staircase-auction}
We define the allocation rule of auction $\bar A$ in the following way:
% \begin{itemize}
     with probability $1/2$ select a bidder $i \in \{1,2\}$
    independently of their bids and then allocate to  $i$ with probability $b_i.$
% \end{itemize}
We define the payment rule in the way that makes the auction truthful.
\end{definition}
A simple application of Myerson's lemma shows that $\gamma_{\bar A} =
\Theta(\nicefrac{1}{\Delta^2}).$ This is because between any two consecutive
bids, i.e., bids whose distance is $1/\Delta,$ the increase in the allocation
is $\nicefrac{1}{2\Delta}$
and the function is linear.
A corollary of \Cref{thm:auction-schedule-regret-upper-bound} shows the following
bound in the auctioneer regret.
\begin{corollary}\label{cor:staircase-auction-regret-bound}
Let the bidders use a mean-based
learner with $\eta_T = \widetilde O(\sqrt{\nicefrac{\log \Delta}{T}})$
and the auctioneer use the schedule 
$(A_1,\ldots,A_T)$ with $A_1 = \ldots = A_{T_0} = \bar A, A_{T_0 + 1} = A_{T_0+2} = \ldots = A_T = \mathrm{SPA},$ for $T_0 = \widetilde O\left(\nicefrac{\sqrt{T}}{\Delta^2}\right)$. Then, $
     \wt{\reg}(A_1,\ldots,A_T;v_L,v_H) \leq \widetilde O\left( \Delta^2\sqrt{T} \right), \forall v_L, v_H \in B^2_\Delta \,.
$
\end{corollary}

\section{Conclusion}\label{sec:conclusion}

Our work studies the behavior of learning bidders
in repeated single-item auctions, with persistent valuations.
We show the limitations of deterministic mechanisms, and how nuances such as learning rates can qualitatively affect participant behavior. Moreover, we show that randomized auctions can encourage faster convergence of bidders to truthful behavior.
We hope our work
paves the way to better understanding of learning agents' behavior 
in single-parameter environments, and of the power of randomization. 

\section*{Acknowledgements}
Anupam Gupta is supported in part  by NSF grants CCF-1955785 and CCF-2006953. Grigoris Velegkas is supported in part
by the AI Institute for Learning-Enabled Optimization at Scale (TILOS).

% Moreover, 
% our approach can be extended to more general settings.

% of this problem and complements the findings of prior work \citep{kolumbus2022auctions},
% by showcasing
% that nuances of the implementation of the learning algorithms, such
% as the ratio of their learning rates, affect the behavior
% of the participants in a \emph{qualitative} way. Our next set of
% results highlights the importance of using \emph{randomized} auctions 
% that speed-up the learning behavior of these algorithms. In particular, 
% we have shown that by carefully balancing between the randomized
% component and the deterministic component of the auctions we can 
% both help the algorithms converge to the desired bidding distribution
% and enjoy the performance guarantees that deterministic auctions exhibit
% when the bidders are \emph{truthful.} 
% We hope our work
% paves the way to better understanding of learning agents' behavior 
% in single-parameter environments, and of the power of randomization. 
% and we expect that the use of randomization will lead
% to results of a similar flavor to the ones in the single-item setting.

%\newpage
% \bibliographystyle{plainnat}
\bibliography{bib}

\newpage
\appendix
%\newpage

\section{Multiplicative Weights Update (MWU)}\label{apx:mwu}
In this section we describe the version of MWU 
we consider in this work. Similar to \citet{braverman2018selling},
we are using the following version of the algorithm.

\begin{algorithm}[ht] 
	\caption{Multiplicative Weights Update Algorithm.} \label{alg:mw}
    \begin{algorithmic}[1]
    	\STATE Choose $\eta_T = \sqrt{\frac{\log \Delta}{ T}}$. Initialize $\Delta$ weights, letting $w^t_{i}$ be the value of the $i$th weight at round $t$. Initially, set all $w^0_{i} = 1$, let $v$ be the valuation of the agent.
        \FOR {$t=1$ to $T$}
			\STATE Choose bid $b_i$ with probability $p^t_{i} = w^{t-1}_{i} / \sum_{j}w^{t-1}_{j}$.
			\FOR {$j=1$ to $K$}
			    \STATE Let $u^t_j = v \cdot x^t(b_j,b') - p^t(b_j,b')$
				\STATE Set $w^t_{j} = w^{t-1}_{j} \cdot e^{\eta_T u^t_{j}}$.
			\ENDFOR
        \ENDFOR
      \end{algorithmic}
\end{algorithm}

\section{Further Related Work}\label{app:related-work}

We view our results and the setting in which we work as orthogonal to the setting of \citet{cai2023selling}. Firstly, they do not restrict themselves to truthful auctions, and for their welfare extraction results, the agents are allowed to overbid. Secondly, in their setting, redrawing valuations i.i.d. in every round helps the learning process (this was also observed by \citet{feng2021convergence}). Intuitively, consider two agents and SPA: for every valuation  of player 1, there is some positive probability that player 2’s draw is below it, hence player 1 will learn that bidding truthfully is strictly better (in expectation over the other random draw), which leads to the desired bidding behavior. In such a system, randomness is already present due to the draws of the valuations, which helps the convergence to the right bidding behavior.

Our work also differs from \citet{cai2023selling} in having different conceptual goals: we aim to ``restore'' the single-shot behavior in natural auctions, such as second-price auctions, in the presence of mean-based learning agents by making minimal modifications to the underlying auction rule. On the other hand, \citet{cai2023selling} aim to exploit the mean-based learning behavior to extract more revenue, and their auctions diverge from the truthful ones we consider in our work. Thus, in our setting, it is clear that reporting the valuation truthfully to the bidding algorithm is an (almost) optimal strategy for the agents (i.e., the so-called ``meta-game'' considered by \citet{kolumbus2022auctions} is truthful), whereas it is not clear to us whether reporting the valuations truthfully to the no-regret algorithms is an optimal strategy in the setting of \citet{cai2023selling}.

\section{Omitted Details from \Cref{sec:setting}}
\citet{SKRETA2006293} shows that our discrete-type space mechanism design problem approximates the mechanism design problem with continuous type space as $\Delta \to \infty$: specifically, Proposition~1 from that paper gives the following claims. 

\begin{claim}
A mechanism is truthful if and only for every $v_{-i}$ $x_i(v_i,v_{-i})$ is non-decreasing on $v_i$ and $p_i$ satisfy that 
$$\left|\,p_i(v_i,v_{-i}) - \left(v_i x_i(v_i,v_{-i}) - \int_0^{v_i} x_i(z,v_{-i}) dz\right)\, \right| \leq O(1/\Delta).$$
\end{claim}
%\grigorisnote{TODO:double-check if we have equality, without the $1/Delta.$}
\begin{claim}
Suppose bidders are rational agents (i.e., they maximize profits). Let $OPT$ be the revenue of the revenue-maximizing mechanism (among truthful or non-truthful) that the auctioneer can implement, and $Rev(r-SPA)$ be the revenue of a Second Price Auction with reserve $r$. Then for $r= \min \{ v: \phi(v) \geq 0 \}$, we have that $OPT = Rev(r-SPA)$. 
\end{claim}

\begin{definition}[No-Regret Learning Property]\label{def:no-regret}
   % Let $n$ be the number of agents and conside some agent $i \in [n].$ For any $t \in [T]$, let $\tilde b^t_{-i} \in B_\Delta^{n-1}$  be an arbitrary sequence of bids submitted by    the agents $j \in [n]\setminus \{i\}.$
    % Let $U_i^T(b \mid\mathbf{b}^T_{-i} ) =  \sum_{\tau=1}^T v_i\cdot x_i^\tau(b,  b^\tau_{-i}) -p_i^\tau(b,  b^\tau_{-i})$ be the cumulative reward agent $i$ gets by using a fixed bid $b$ over the $[T]$ rounds and the other agent bids according to $\mathbf{b}^T_{-i} = \{b^\tau_{-i}\}_{\tau \in [T]}$. 
    Let $\{b_i^\tau\}_{\tau \in [T]}$ be the bid
    sequence submitted by agent $i$'s algorithm, and $U_i^T(\mathbf{b}^T ) =  \sum_{\tau=1}^T v_i\cdot x_i^\tau(b_i^\tau,  b^\tau_{-i}) -p_i^\tau(b_i^\tau,  b^\tau_{-i})$
    the total reward agent $i$ receives.
    %we denote by $U_i^t(b |\mathbf{b}^t_{-i} ) =  \sum_{\tau=1}^t v_i\cdot x_i^\tau(b, \tilde b^\tau_{-i}) -p_i^\tau(b, \tilde b^\tau_{-i})$ the cumulative reward agent $i$ gets by using a fixed bid $b$ over the $[t]$ rounds.
    %\andresnote{this should be vector of bids for $t$?}\grigorisnote{I changed the notation slightly, does it make more sense now? $b_{-i}$ is vector of bids, if you prefer to use bold for it I'm fine with it. If we stick to this notation, I will also change the no-regret definition.}\andresnote{let's use bold for vectors :)}.
    We say that this algorithm
    \emph{satisfies the no-regret property} if for any
    sequence $\mathbf{b}^T_{-i}$ it holds that \[
        \E\left[\max_{b \in B_\Delta} U_i^T(b\mid\mathbf{b}_{-i}^T) - U_i^T(\mathbf{b}^T) \right] = o(T) \,,
    \]
    where the expectation is taken
    with respect to the randomness 
    of the algorithm.
   % \andresnote{is the def. still correct?}\grigorisnote{Looks good to me!}
\end{definition}

\begin{definition}[Last Iterate Convergence (LIC)]\label{def:last-iterate-convergence}
    Let $\tildeb_i^T$ the bid distribution of bidder $i$ in the last round $T$. We say that $\tildeb_i^T$ converges to some distribution $\tilde q$ over $B_\Delta$ if $
\lim_{T\to \infty} d_{\mathrm{TV}}(\tildeb_i^T,\tilde q) = o(1),$
where $d_{\mathrm{TV}} := \frac{1}{2}\left(\sum_{b \in B_\Delta}|\tildeb_i^T(b) - \tilde q(b)|\right)$
is the Total-Variation (TV) distance between $\tildeb_i^T$ and $\tilde q.$
\end{definition}

\section{Omitted Details from \Cref{sec:deterministic auctions}}\label{app:deterministic-auctions}
\begin{definition}[Non-Degenerate auctions]\label{def:non-degenerate}
A single-item auction $(x, p)$ for two agents 
is non-degenerate with respect to the
valuation profile $(v_1, v_2)$ if
there are bid profiles $b_1 \leq v_1, b_2 \leq v_2,$ so that
\begin{align*}
    v_1 \cdot x_1(v_1, b_2) - p_1(v_1,b_2) &>  v_1 \cdot x_1(v_1-\nicefrac{1}{\Delta}, b_2) - p_1(v_1,b_2) \geq 0 \\
     v_2 \cdot x_2(b_1, v_2) - p_2(b_1,v_2) &>  v_2 \cdot x_2(b_1, v_2-\nicefrac{1}{\Delta}) - p_2(b_1,v_2-\nicefrac{1}{\Delta}) \geq 0  \,,
\end{align*}
and
\begin{align*}
    \max\left\{v_1 \cdot x_1(v_1, v_2) - p_1(v_1,v_2), v_2 \cdot x_2(v_1, v_2) - p_2(v_1,v_2) \right\} > 0\,.
\end{align*}
\end{definition}

In order to show our result, we utilize 
a characterization (cf. \Cref{thm:characterization-deterministic-auctions}) regarding the structure of truthful
deterministic
single-item auctions that charge non-negative payments (see, e.g., \citet[Thm~9.36]{roughgarden2010algorithmic}) for $n$ bidders. 

\begin{theorem}[Characterization of Truthful Deterministic Single-Item Auctions \cite{roughgarden2010algorithmic}]\label{thm:characterization-deterministic-auctions}
A single-item auction is truthful, and satisfies NPT, i.e., no payment transfers from the auctioneer to the bidders, if and only if:
\begin{itemize}
    \item $x_i(\cdot,v_{-i})$ is monotone for every $i \in [n], v_{-i} \in B_\Delta^{n-1}.$
    \item For all $i \in [n], v_i \in B_\Delta, v_{-i} \in B_\Delta^{n-1}$ we have that
    \begin{align*}
        p_i(v_i,v_{-i}) = \begin{cases}
        0, & \text{ if } x_i(v_i,v_{-i}) = 0 \\
    \min \{b \in B_\Delta: x_i(b,v_{-i}) = 1\}, & 
    \text{ if }  x_i(v_i,v_{-i}) = 1
    \end{cases}  \,.
    \end{align*}
\end{itemize}

\end{theorem}

\begin{theorem}[No Deterministic Auction Leads to Truthful Bidding]\label{theorem:deterministic anonymous auctions do not lead to convergence}
Fix a valuation profile $(v_1, v_2)$ and a  deterministic truthful auction.
Suppose bidders bid using MWU and with non-degenerate
learning rates. Let  $W$ (respectively $R$), be the bidder $i \in \{1,2\}$ such that $x_i(v_i,v_{-i}) = 1$ (respectively, $x_i(v_i,v_{-i}) = 0$) and let $\hat{p} = p_W(v_W, v_R).$ Assume that $\lim_{T \rightarrow \infty}\nicefrac{\eta_T^R}{\eta_T^W} < \infty$
and $v_W \cdot x_W(v_W,v_R) - \hat p > 0.$
Then, with probability at least $0.99$,
% where $C_\Delta^1, C_\Delta^2 > 0,$ are some constants that depend on $\Delta$,
the winner's bids converge to a distribution supported
between $\hat p, v_W$ and the runner-up bidder converges to a bidding distribution satisfying $0 < \Pr[0] \leq \Pr[1/\Delta] \leq \ldots \leq \Pr[v_R].$
\end{theorem}

\begin{proof}[Proof of \Cref{theorem:deterministic anonymous auctions do not lead to convergence}]
 The idea of the proof
is to split the horizon $T$ into continuous non-overlapping epochs of length
$c/\eta_T^W$, where $c$ is some sufficiently large constant that depends
on the discretization parameter $\Delta$. Notice that since $\lim_{T \rightarrow \infty}\eta_T^W \cdot T = \infty$ these epochs are well-defined, when $T$ is sufficiently large.
Assume without loss of generality
that the weights of all the bids that are at most $v_W$ (resp. $v_R$)
for the winning bidder (resp. runner-up) are initialized
to 1. (The proof holds as long as there is some constant mass on each
bid at the initialization stage, albeit with different constants.)
We denote the epochs by $\tau$ and the rounds of
the interaction by $t.$

Let $c_W = v_W - \hat{p}$ be the utility the bidder gets
when it wins the auction. By assumption, $c_W > 0.$ Let $W_W$ be the set of bids between $\hat p$ and  $v_W$, i.e., $W_W = \{\hat p,\hat p + \nicefrac{1}{\Delta},\ldots, v_W \}.$
Whenever the runner-up bids $v_R$ all the bids in $W_W$
increase their weights by a multiplicative factor of $e^{
c_W\cdot\eta_T^W}$, whereas the weights of the other bids remain unchanged. 
Moreover, since the allocation rule is non-decreasing and the price
does not depend on the bid,
whenever the weight of some bid $b \in B_\Delta$ 
is increased, the weights of all the bids
that are greater than $b$
are also increased
by the same amount.
Notice that, since
bidding $v_R$ is a weakly-dominant strategy
for the runner-up type, the mass that it puts on $v_R$
will never decrease relatively to the mass of the
rest of the bids. Thus, the probability of bidding $v_R$ for the runner-up
type is at least $1/\Delta$ in every round.
Hence, if we consider an interval
of size $T_0 = 8\Delta^2 / (\eta_T^W \cdot c_W)$ 
and we denote by $Z_i, i \in [T_0],$ the indicator variable
of whether the runner-up bid $v_R$ in round $i \in [T_0]$ 
we have that for any $\alpha > 0$
\[
    \Pr\left[Z_1 + \ldots + Z_{T_0} \geq \alpha \right] \geq 
    \Pr\left[\tilde Z_1 + \ldots + \tilde Z_{T_0} \geq \alpha \right] \,,
\]
where $\tilde Z_i, \in [T_0]$ are i.i.d. Bernoulli random 
variables with mean $1/\Delta.$
Then, the multiplicative 
version of Chernoff bound on $\{\tilde Z_i\}_{i \in [T_0]}$
shows that, 
with probability at least $1-e^{-\Delta / (\eta_T^W \cdot c_W)}$
the runner-up type will bid at least $4 \Delta / (\eta_T^W \cdot c_W))$
many times
$v_R$ in this window. By a union bound, we know that with 
probability at least $1-(T\cdot \eta_T^W/c) \cdot e^{-\Delta / (\eta_T^W \cdot c_W)}$
this holds across all the $T\cdot \eta_T^W/c$ different epochs. We call this 
event $\calE_1$ and condition on it for the rest of the proof.
Our assumption that $\eta_T$ is non-degenerate shows
that this probability is at least $1-o(1).$

Let $w_W^\tau(b)$ be the total weight that the winning type
assigns to $b$ at the beginning of epoch $\tau$
and $m_W^\tau(b)$ be its probability. Notice
that at $\tau=1$ this distribution is uniform.
Consider the ratio of the weights of any $b \leq \hat p - 1/\Delta$
and $\hat p.$ We have that
\begin{align}\label{eq:1}
    \frac{w_W^{\tau+1}(b)}{w_W^{\tau+1}(\hat p)} &\leq \frac{w_W^{\tau}(b)}{w_W^{\tau}(\hat p)}\cdot e^{- 4c_W\cdot \Delta\cdot \eta_T^W/(c_W \cdot \eta_T^W)} = \frac{w_W^{\tau}(b)}{w_W^{\tau}(\hat p)}\cdot e^{-4\Delta} \,,
    % \implies \\
    % m_W^{\tau+1}(b) &\leq \frac{ m_W^{\tau+1}(\hat p)}{ m_W^{\tau}(\hat p)} \cdot m_W^\tau(b) \cdot e^{-4\Delta} \,,
\end{align}
 where $w_W^{\tau}(b), w_W^{\tau}(\hat p)$ are the weights that the winner puts on $b, \hat p$ at the beginning of epoch $\tau$ (similarly for the $\tau+1$ terms). 
   % We can equivalently write \Cref{eq:1} as
   %  \begin{equation}\label{eq:2}
   %      \frac{\frac{w_W^{\tau+1}(b)}{\sum_{b' \in B_\Delta} w_W^{\tau+1}(b')}}{\frac{w_W^{\tau+1}(\hat p)}{\sum_{b' \in B_\Delta} w_W^{\tau+1}(b')}} \leq \frac{\frac{w_W^{\tau}(b)}{\sum_{b' \in B_\Delta} w_W^{\tau}(b')}}{\frac{w_W^{\tau}(\hat p)}{\sum_{b' \in B_\Delta} w_W^{\tau}(b')}}\cdot e^{-4\Delta}
   %  \end{equation}
    For the probability of each
    bid in MWU, 
    $m_W^{\tau + 1}(b) = \frac{w_W^{\tau+1}(b)}{\sum_{b' \in B_\Delta} w_W^{\tau+1}(b')}$ (and symmetrically for the other terms). Thus, by dividing the numerator and the denominator of the RHS of \Cref{eq:1}
    by $\sum_{b' \in B_\Delta} w_W^{\tau}(b')$
    and the numerator and denominator of the LHS of \Cref{eq:1} by $\sum_{b' \in B_\Delta} w_W^{\tau+1}(b')$
    we get:
    \[
        \frac{m_W^{\tau+1}(b)}{m_W^{\tau+1}(\hat p)} \leq \frac{m_W^{\tau}(b)}{m_W^{\tau}(\hat p)}\cdot e^{-4\Delta}.
    \]
    Multiplying by $m_W^{\tau+1}(\hat p)$ gives us
    \[
         m_W^{\tau+1}(b) \leq \frac{ m_W^{\tau+1}(\hat p)}{ m_W^{\tau}(\hat p)} \cdot m_W^\tau(b) \cdot e^{-4\Delta} \,.
    \]
Notice that $m_W^1(\hat p) = 1/\Delta, m_W^\tau(\hat p)$ is
non-decreasing in $\tau$ since bidding $\hat p$
is a weakly-dominant strategy for the winning type\footnote{This is where we are using the assumption that the runner-up type does not overbid. Otherwise, the argument can still go through with a different constant since we can show that the winning type will overbid only some $O(\eta_T^W)$ many times and we need to account for this term.},
and, by definition, $m_H^{\tau+1}(\hat p) \leq 1$,
so $\frac{ m_W^{\tau+1}(\hat p)}{ m_W^{\tau}(\hat p)} \leq \Delta.$ Hence,
\[
    m_W^{\tau+1}(b) \leq \Delta e^{-4\Delta} \cdot m_W^{\tau+1}(b)< 0.1 \cdot m_W^{\tau}(b), \forall b < \hat p \,,
\]
where the second inequality follows from $xe^{-4x} < 1, \forall x > 0.$
Thus, after each epoch the probability that the winning type does not
bid in $W_W$ decreases by a factor of $0.9.$
Hence, we can see that after $O(\eta_T^W \cdot T)$ epochs
that total mass in this region is at most $O(0.1^{\eta_T^W \cdot T-1}) = o(1).$ This proves the claim about 
the distribution of the winning type.

Let $Z_i, i \in [T],$ be the random variable
that indicates whether $v_W$ bid in 
$\{0,1/\Delta,\ldots,\hat p-1/\Delta\}$
in round $i \in [T].$ Let also $T'$ denote
the total number of epochs.
 Let $\widehat Z_\tau = Z_{\tau} + \ldots + Z_{\tau + T_0-1}$, so that 
    % By linearity of expectation we have
    $\mathbb{E}[Z_1 + \ldots Z_T] = \sum_{\tau=1}^{T'} \mathbb{E}[\widehat Z_\tau].$ 
    The preceding steps of the proof had 
    shown that after every round,
    the probability that the winner bids in this region
    is non-increasing (since the bids in interval $I$ are
    weakly dominated by the bids in $\{\widehat p,\ldots,v_W\}$), hence $\mathbb{E}[\widehat Z_\tau] \leq T_0\cdot  \mathbb{E}[ Z_{(\tau-1)\cdot T_0 + 1}].$ Thus, it suffices to bound
    $\sum_{\tau = 1}^{T'}\mathbb{E}[ Z_{(\tau-1)\cdot T_0 + 1}].$ 
    
    By definition, $\mathbb{E}[ Z_{(\tau-1)\cdot T_0 + 1}] = \sum_{b < \widehat p} m_W^{(\tau-1)\cdot T_0 + 1}(b).$ Now, the 
    previous step of the proof had shown
    that the mass
    of each bid in interval $I$ drops by a factor of 0.9 between the beginning
    of consecutive epochs, i.e., 
    $m_W^{\tau\cdot T_0 + 1}(b) \leq 0.1\cdot m_W^{(\tau-1)\cdot T_0 + 1}(b)$ for all $b \in \{0,1/\Delta,\ldots \widehat{p} - 1\}$. This implies $\mathbb{E}[ Z_{\tau\cdot T_0 + 1}] \leq 0.1 \cdot \mathbb{E}[ Z_{(\tau-1)\cdot T_0 + 1}].$
    Using $\mathbb{E}[ Z_{ 1}] \leq 1$, we get $\sum_{\tau = 1}^{T'}\mathbb{E}[ Z_{(\tau-1)\cdot T_0 + 1}] \leq \sum_{\tau = 1}^{T'} (0.1)^{\tau-1}$. Multiplying this by the value of $T_0$ gives
\begin{align*}
    \E\left[Z_1 + \ldots + Z_T\right] &\leq \sum_{\tau = 1}^{T'} (8\Delta^2 / (\eta_T^W \cdot c_W)) \cdot (0.1)^{\tau - 1} \\
    &\leq \sum_{\tau = 1}^{\infty} (8\Delta^2 / (\eta_T^W \cdot c_W)) \cdot (0.1)^{\tau - 1} \\
    &\leq 16\Delta^2 / (\eta_T^W \cdot c_W) \,.
\end{align*}
Hence, using Markov's inequality we see that
\begin{align*}
    \Pr\left[Z_1 + \ldots + Z_T \geq 101\cdot\left(16\Delta^2 / (\eta_T^W \cdot c_W) \right)\right] &\leq \frac{\E\left[Z_1 + \ldots Z_T\right]}{101\cdot\left(16\Delta^2 / (\eta_T^W \cdot c_W) \right)} 
    \leq \frac{1}{101} \,.
\end{align*}
Let us call this event $\calE_2$ and condition on it.

% In the first epoch, the winning type bids in the  region $\{0,1/\Delta,\ldots,\hat p-1/\Delta\}$
% at most $(c /\eta_T^W)\cdot(\hat p\cdot \Delta)/(\Delta + 1)$ times in expectation, since the total probability mass in this region is initialized to $(\hat p\cdot \Delta)/(\Delta + 1)$,
% and it does not increase throughout 
% the execution, since bidding in $W_W$ is a weakly-dominant
% strategy.
% Again, a multiplicative Chernoff bound shows that
% with probability at least $1-e^{-(c /\eta_T^W)\cdot(\hat p\cdot \Delta)/(\Delta + 1)}$
% the winning type will not bid more than $2 (c /\eta_T^W)\cdot(\hat p\cdot \Delta)/(\Delta + 1)$
% times in this region. As we argued above, the expected
% number of times it bids in this region in the second epoch 
% shrinks by a factor of $0.9.$
% \grigorisnote{TODO:fix concentration bound.}
% Hence, by taking a union bound
% over the concentration for each epoch and summing up the geometric series,
% we see that with probability
% $1-(T\cdot \eta_T^W/c) \cdot e^{-(c /\eta_T^W)\cdot(\hat p\cdot \Delta)/(\Delta + 1)} = 1 - o(1)$
% the winning type will bid in this region at most $(4 c /\eta_T^W)\cdot(\hat p\cdot \Delta)/(\Delta + 1).$
% We call this event $\calE_2$ and we condition on it.

Let us now consider the bid distribution of the runner-up type after the end 
of the last epoch. We denote this distribution by $\wh m_R(\cdot)$.
Recall that whenever the winning type bids in $W_W$, the runner-up type
performs no updates. Moreover, whenever it does perform an update
its utility when it bids $v_R$ is at most $1$ greater than bidding $b = 0.$
Notice that whenever the weight of some bid $b$ is increased, the weights
of all the bids greater than $b$ are also increased by the same amount,
so the monotonicity of the bid distribution follows immediately.
It suffices now to bound the ratio of the probability of bidding
zero and the probability of bidding $v_R$ by some quantity that
is independent of $T.$ We have that
\[
    \frac{\wh m_R(0)}{\wh m_R(v_R)} \geq e^{-\eta_T^R 101\cdot\left(16\Delta^2 / (\eta_T^W \cdot c_W) \right)} \implies \\
    \wh m_R(0) \geq \frac{ e^{-\eta_T^R 101\cdot\left(16\Delta^2 / (\eta_T^W \cdot c_W) \right)}}{\Delta} \,,
\]
where the second inequality follows from the fact that
the distribution is initialized to be uniform and
$v_R$ is a weakly-dominant strategy across all rounds,
so its probability is not decreased.
Notice that
\[
    \lim_{T \rightarrow \infty} \nicefrac{\eta_T^R}{\eta_T^W} < C\,,
\]
for some discretization-dependent $C$, it follows that 
$\wh m_R(0) > C',$ where $C' > 0$ is some discretization-dependent
constant. Since $\Pr[\calE_1] \geq 1 - o(1), \Pr[\calE_2] \geq 100/101,$
we have that $\Pr[\calE_1 \cap \calE_2] \geq 99/100,$ when $T$ is large
enough.
\end{proof}

\begin{theorem}[Effect of Learning Rate on Convergence]\label{prop:effect-of-learning-rate-convergence}
Fix a valuation profile $(v_1, v_2)$ and a  non-degenerate deterministic truthful auction with respect to $(v_1,v_2).$ 
Suppose bidders bid using MWU and with non-degenerate
learning rates. Let $W$ (respectively $R$), be the bidder $i \in \{1,2\}$ such that $x_i(v_i,v_{-i}) = 1$ (respectively, $x_i(v_i,v_{-i}) = 0$). 
Let $\hat p$ be the minimum winning bid of $W$
when $R$ bids $v_R.$
Assume that $\nicefrac{\eta_T^R}{\eta_T^W} = \omega(1).$
Then, with probability at least $1 - o(1)$,
% where $C_\Delta^1, C_\Delta^2 > 0,$ are some constants that depend on $\Delta$,
bidder $R$ converges to bidding $v_R$ 
and bidder $W$ converges to a bidding distribution
supported in $\{\hat p, \hat p + \nicefrac{1}{\Delta},\ldots, v_W\}.$
\end{theorem}
\begin{proof}[Proof of \Cref{prop:effect-of-learning-rate-convergence}]
    Consider the first $T_0 = c_\Delta'/ \eta_T^W$
    rounds of the game, for some $c_\Delta'$ discretization-dependent
    constant. Assume without loss of generality
that the weights of all the bids that are at most $v_W$ (resp. $v_R$)
for the winning bidder (resp. runner-up) are initialized
to 1. (Again, the argument works so long as all the weights
are initialized with some constants.) 
Since the auction is non-degenerate with respect
to $v_W, v_R,$ there exists some bid of the winning type $b_W \leq v_W$ so that
the runner-up bidder wins the auction when bidding truthfully and gets positive utility, i.e.,
\[
   v_R \cdot x_R(v_R, b_W) - p_R(v_R, b_W)  > 0 \,.
\]
Moreover, for all bids $b_R < v_R$ it holds
\[
    v_R \cdot x_R(v_R, b_W) - p_R(v_R, b_W) -
    \left(v_R \cdot x_R(b_R, b_W) - p_R(b_R, b_W)\right) > 0\,.
\]
Since the auction is truthful, the difference
above is minimized at $b_R = v_R - \nicefrac{1}{\Delta}.$ Let
\[
    u'_R := v_R \cdot x_R(v_R, b_W) - p_R(v_R, b_W) -
    \left(v_R \cdot x_R(v_R-\nicefrac{1}{\Delta}, b_W) - p_R(v_R-\nicefrac{1}{\Delta}, b_W)\right) \,,
\]
and, by definition, $u'_R > 0.$
    Let us consider the winning type and look at the worst-case ratio
    of the probability that is placed on bids $b_W^{t} = b_W, b_W^{t} = v_W$
    at the end of every round $t \in \{1,\ldots,T_0\}$. We have that
    \begin{align*}
        \frac{\Pr[b_W^{t} =b_W]}{\Pr[b_W^{t} = v_W]} 
        &\geq e^{- \eta_T^W\cdot v_W\cdot t} \\
        &\geq e^{-\eta_T^W\cdot v_W\cdot T_0} \\
        &= e^{- c_\Delta' \cdot v_W} \,,
    \end{align*}
    where the first inequality follows from the fact that bidding $v_W$
    always yields at most $v_W$ utility more than bidding
    any other bid and the second one because $t \leq T_0$. Moreover, since $\Pr[b_W^{1} = v_W] = 1/\Delta$
    and the probability that is placed on $b_W^t = v_W$ is non-decreasing
    across the executions (since it is a weakly-dominant strategy),
    we have that
    \[
        \Pr[b_W^{t} = b_W] \geq  e^{- c_\Delta' \cdot v_W}/\Delta, \forall t \in \{1,\ldots,T_0\} \,.
    \]
    
    Let $Z^{T_0}$ denote the random variable that counts 
    the number of times the winning type bids $b_W$
    within the first $T_0$
    rounds. Let $\tilde Z_{\tau}, \tau \in [T_0]$
    be independent Bernoulli random variables with mean $e^{- c_\Delta' \cdot v_W}/\Delta$.
    Notice that, $\forall \alpha > 0,$ it holds 
    that $\Pr[Z^{T_0} \geq \alpha] \geq \Pr[\sum_{\tau = 1}^{T_0}
    \tilde Z_\tau  \geq \alpha].$ Moreover,
    \[
        \E\left[\sum_{\tau = 1}^{T_0}
    \tilde Z_\tau \right] \geq T_0 \cdot  e^{- c_\Delta' \cdot v_W}/\Delta = c_\Delta'/ \eta_T^W \cdot e^{- c_\Delta' \cdot v_W}/\Delta \,.
    \]
    To simplify the notation, let us denote 
    $\tilde{c}_\Delta = c_\Delta' \cdot e^{- c_\Delta' \cdot v_H}/\Delta.$
    Thus, a multiplicative Chernoff bound shows that, with probability
    at least $1- e^{- \tilde{c}_\Delta/(8\eta_T^W)} = 1 - o(1),$
    we have that $Z^{T_0} \geq \tilde{c}_\Delta/(2\eta_T^W).$
    Let us call this event $E$ and condition on it.
    
    Let us now focus on the bid distribution of the runner-up
    bidder after the first $T_0$ rounds. Notice that
    whenever the winning bidder bids $b_W$ then
    bidding $v_R$ yields utility at least $u'_R$ greater than bidding any other bid to the runner-up type,
    and
    in the rounds where this does not happen, bidding
    $v_R$ is still a weakly dominant strategy so it generates
    as much utility as any other bid. Thus,
    we have that
    \begin{align*}
        \frac{\Pr[b_R^{T_0} = v_R-1/\Delta]}{\Pr[b_R^{T_0} = v_R]} &\leq  e^{-u'_R\cdot \eta_T^R \cdot Z^{T_0}}\\
        &\leq e^{-\eta_T^R\cdot \tilde{c}_\Delta/(2\eta_T^W\Delta)} \\
        & = o(1)
    \end{align*}
    Thus, since bidding $v_R$ is a weakly dominant
    strategy for the runner-up this ratio is non-increasing in $t$
    we can immediately see that
    \[
         \frac{\Pr[b_R^{T_0} = v_R-1/\Delta]}{\Pr[b_R^{T_0} = v_R]} = o(1) \,,
    \]
    which gives that
    \[
        \Pr[b_R^{T} = v_R-1/\Delta] = o(1) \,.
    \]
    The same argument can be applied to all bids in $\{0,1/\Delta,\ldots,v_R - 1/\Delta\}.$
    
    For the winning type, a symmetric argument shows that
    since after $O(\eta_T^W)$ many rounds the runner-up type bids
    $v_R$ with high probability, all the bids in the region
    $\{\hat v_W,\ldots,v_W\}$ will yield utility that is 
    larger than bidding $v_R - 1/\Delta$
    by at least $1/\Delta$ (again with
    high probability), so after another $\omega(\eta_T^W)$ rounds
    its mass will be concentrated on bidding in this region.
\end{proof}

\begin{proof}[Proof of \Cref{thm:deterministic-auctions-not-revenue-optimal}]
Let $\calE = \{r< v_1\} \cap \{r < v_2\} \cap\{ v_1 \neq v_2\}.$
We can decompose $\E_{v_1,v_2 \sim U[B_\Delta]}\left[\mathrm{Rev}(v_1,v_2;r)\right]$ as:
\begin{align*}
    \E_{v_1,v_2 \sim U[B_\Delta]}\left[\mathrm{Rev}(v_1,v_2;r)\right] = \E_{v_1,v_2 \sim U[B_\Delta]}\left[\mathrm{Rev}(v_1,v_2;r)\right| \calE] \cdot \Pr_{v_1,v_2 \sim U[B_\Delta]}[\calE] \\
    + \E_{v_1,v_2 \sim U[B_\Delta]}\left[\mathrm{Rev}(v_1,v_2;r)\right| \calE'] \cdot \Pr_{v_1,v_2 \sim U[B_\Delta]}[\calE']   \,.
\end{align*}

Notice that under $\calE',$ the revenue of the auction in the learning 
setting satisfies
\[
    \E_{v_1,v_2 \sim U[B_\Delta]}\left[\lim_{T \rightarrow \infty}\E_{b_1 \sim b_1^T, b_2 \sim b_2^T}[\mathrm{Rev}(b_1,b_2;r) \mid v_1, v_2 ]\,\bigg\vert\, \calE'\right] \leq \E_{v_1,v_2 \sim U[B_\Delta]}\left[\mathrm{Rev}(v_1,v_2;r)\right| \calE'] \,.
\]
This is because both bidders will be bidding at most their valuation,
so the revenue of the auction cannot increase. Let us now focus
on the first term. Under the event $\calE,$ the revenue of the auction
under rational agents is $\min\{v_1,v_2\} > r.$
However, in the learning setting, the runner-up bidder
will be bidding strictly below their valuation in expectation, by \Cref{theorem:deterministic anonymous auctions do not lead to convergence}. Hence, we have that
\begin{align*}
    \E_{v_1,v_2 \sim U[B_\Delta]}\left[\lim_{T \rightarrow \infty}\E_{b_1 \sim b_1^T, b_2 \sim b_2^T}[\mathrm{Rev}(b_1,b_2;r) \mid v_1, v_2 ] \,\bigg\vert\, \calE \right] &< \E_{v_1,v_2 \sim U[B_\Delta]}\left[ \min\{v_1, v_2\} \,\bigg\vert\, \calE\right] - c' \\
    &= \E_{v_1,v_2 \sim U[B_\Delta]}\left[\mathrm{Rev}(v_1,v_2;r)\right| \calE] - c' \,.
\end{align*}
Since $\Pr[\calE] > 0,$ the result follows by combining the two inequalities.
\end{proof}

\section{Omitted Details from \Cref{sec:randomized-truthful-auctions}}
\label{app:randomized-asymptotic}

\begin{proof}[Proof of \Cref{lem:convergence-strictly-IC-auctions}]
Let
\[
    \gamma_A = \min_{i \in [n], v \in B_\Delta, b_{-i} \in B_\Delta^{n-1}, b \in B_\Delta: b\neq v}\{u_i(v,b_{-i}) - u_i(b, b_{-i} )\} \,,
\]
i.e., the minimum improvement in the utility that
is guaranteed to every player when they 
switch to bidding truthfully from any non-truthful bid,
no matter what their valuation
and the bids of the opponents are.
Notice that for any fixed auction $A$ this quantity
does not depend on $T.$
Moreover, since $A$ is a strictly IC auction
we have that $\gamma_A > 0.$
Consider any round $t \in [T]$ of the interaction.
For any player $i \in [n],$ we have that
\[
    u^t(v_i,b^t_{-i}) - u^t(b',b^t_{-i}) \geq \gamma_A, \forall b' \neq v_i \,,
\]
no matter what the bids $b^{t}_{-i}$ are.
Let $\delta_1,\ldots,\delta_n$ be the mean-based
parameters of the algorithms that the agents
are using. Moreover, let $T_0 = \max_{i \in [n]} \delta_i\cdot T/\gamma_A.$ Notice that since
$\delta_i = o(1), \forall i \in [n],$
by picking $T$ sufficiently large we have
that $T_0 < T.$
We immediately get that, for every
player $i \in [n]$
\[
    \sum_{t=1}^{T_0}\left( u^t(v_i,b^t_{-i}) - u^t(b',b^t_{-i})\right) \geq \gamma_A \cdot T_0 \geq \delta_i \cdot T , \forall b' \neq v_i \,,
\]
no matter what the bid profile $b^t_{-i}$
of the other bidders in every round is.
Thus, for every bidder $i \in [n],$
by taking a union bound over all bids $b \neq v_i$,
we see that 
in round $T_0 + 1$ the probability of not
bidding truthfully is at most $\Delta\cdot \delta_i = o(1).$ Hence, we have shown the result.
\end{proof}
% \begin{theorem}\label{thm:mixture-of-IC-with-strictly-IC-last-iterate-convergence}
% Let $A$ be an IC auction for $n$ agents with
% valuations $v_1, \ldots, v_n$. Let each agent $i \in [n]$
% use a mean-based no-regret learning algorithm to bid in the auction.
% Then, there exists an auction $A'$ such that 
% for each agent $i \in [n]$
% we have that $\lim_{T \rightarrow \infty} b^T_i =  v_i$ and $|x_i(b) - x'_i(b)| = o(1), |p_i(b) - p'_i(b)| = o(1),\forall b\in B^n_\Delta,$ where $x_i(\cdot), x'_i(\cdot)$ (resp. $p_i(\cdot), p'_i(\cdot)$) is the allocation (resp. payment) rule of $A, A'.$
% \end{theorem}

\begin{proof}[Proof of \Cref{thm:mixture-of-IC-with-strictly-IC-last-iterate-convergence}]
    Let $\delta_,\ldots,\delta_n$ be the mean-based
    parameters of the algorithms that the agents
    are using. Recall that these parameters do depend
    on $T.$
    Assume without loss of generality that
    $\delta_1$ is the slowest one, i.e., $\lim_{T \rightarrow \infty} \nicefrac{\delta_i}{\delta_1} \leq C, \forall i \in [n],$ where
    $C$ is some discretization-dependent constant.
    Let $\widetilde{A}$ be a strictly IC
    auction and define
    \[
        \gamma_{\widetilde{A}} = \min_{i \in [n], v \in B_\Delta, b_{-i} \in B_\Delta^{n-1}, b \in B_\Delta :b \neq v}\{\widetilde u_i(v,b_{-i}) - \widetilde u_i(b, b_{-i})\} \,.
    \]
    Similarly as in the previous proof,
    notice that $\gamma_{\widetilde{A}}$ does not depend on $T.$ Consider the $q_T$-mixture of the auctions
    $A, \widetilde{A}$ and let us denote this auction
    by $A'.$ Let $x, \widetilde{x}, x'$ be the allocation
    rules of $A, \widetilde{A}, A',$ respectively, and
    let us define the payment rules in a symmetric way.
    Notice that since $x'(\cdot) = q_T \widetilde{x}(\cdot) + (1-q_T)x(\cdot), p'(\cdot) = q_T \widetilde{p}(\cdot) + (1-q_T)p(\cdot)$, it
    follows immediately that
    \[
        \gamma_{A'} \geq q_T\cdot \gamma_{\widetilde{A}} \,.
    \]
    Moreover, notice that
    \begin{align*}
        |x'(\cdot) - x(\cdot)| &\leq q_T\cdot |\widetilde{x}(\cdot) - x(\cdot)| \leq q_T\\
        |p'(\cdot) - p(\cdot)| &\leq q_T\cdot |\widetilde{p}(\cdot) - p(\cdot)| \leq q_T \,.
    \end{align*}
    
    Let us focus on agent 1 since it is the one
    that has the slowest convergence. After $T_0$
    rounds of the game we have that
    \[
    \sum_{t=1}^{T_0}\left( u^t(v_1,b^t_{-1}) - u^t(b',b^t_{-1})\right) \geq \gamma_{A'} \cdot T_0 \geq q_T\cdot \gamma_{\widetilde{A}} \cdot T_0 , \forall b' \neq v_1 \,,
\]
    no matter what the bid profile of the rest of 
    the bidders in every round is. Thus, in order
    for the mean-based guarantee of the algorithm
    of the first bidder to give us the desired
    convergence
    we see
    that we need $T_0 \geq \nicefrac{\delta_1 \cdot T}{q_T \cdot \gamma_{\widetilde{A}}}.$ Since $T_0 \leq T,$
    this places a constraint on the choice of $q_T$, 
    namely that $q_T \geq \nicefrac{\delta_1}{\gamma_{\widetilde{A}}}.$
    Thus, since this is the only constraint that 
    we have on the choice of $q_T$ we see that
    choosing $q_T = \nicefrac{2 \delta_1}{\gamma_{\widetilde{A}}} = o(1)$ suffices to get
    the result.
\end{proof}

\begin{proof}[Proof of \Cref{cor:revenue-guarantees-randomized}]
    Let $A'$ be the output of \Cref{thm:mixture-of-IC-with-strictly-IC-last-iterate-convergence} when
    the input auction is Myerson's revenue-optimal auction for $F.$
    For any fixed valuation profile $v \in B_\Delta^n$, for sufficiently 
    large $T,$ each bidder $i \in [n]$ will be bidding $v_i$ except
    with probability $o(1).$ Moreover, the payments in these two auctions
    differ by some $o(1)$. Thus, 
    \[
    \E_{b_1\sim b^T_1,\ldots,b_n \sim b^T_n}\left[\lim_{T \rightarrow \infty}\mathrm{Rev}(A;b_1,\ldots,b_n) \right] \geq \mathrm{Rev}(\mathrm{Myerson};v_1,\ldots,v_n) - o(1) \,.
    \]
    The result follows by taking the expectation over the random draw of $v_1,\ldots,v_n.$
\end{proof}

We present the formal result about the equilibria
of the meta-game below.

\begin{corollary}[Equilibria of Meta-Game]\label{cor:equilibria of meta-game}
Let $A$ be an IC, IR auction. Let $T$ be the number of interactions.
    Assume that $n$ agents use mean-based
    no-regret learning 
    algorithms to bid in these repeated auctions. 
    Then, there is an auction $A'$ such that
    \begin{itemize}
        \item $|x_i(b) - x'_i(b)| = o(1), |p_i(b) -  p'_i(b)| = o(1), \forall i \in [n], \forall b \in B^n_\Delta.$
        \item In the meta-game that is induced
        by $A'$ every agent
        can gain at most $o(1)$ utility
        by misreporting its value to the bidding
        algorithm.
    \end{itemize}
\end{corollary}

\begin{proof}[Proof of \Cref{cor:equilibria of meta-game}]
    Let $v_1,\ldots,v_n$ be the values of the agents
    and let $\hat v_1, \ldots, \hat v_n$ be the reports
    to the bidding algorithms.
    Let $A'$ be auction obtained by feeding the 
    auction $A$ into the transformation described
    in \Cref{thm:mixture-of-IC-with-strictly-IC-last-iterate-convergence}. The guarantees of this result
    show that 
    \begin{itemize}
        \item $|x_i(b) - x'_i(b)| = o(1), |p_i(b) -  p'_i(b)| = o(1), \forall i \in [n]\forall b \in B_\Delta,$
        \item $\Pr[b_i^T \neq \hat v_i] = o(1),\forall i \in [n],$
    \end{itemize}
    where $b_i^T$ is the bid of the $i$-th agent in round $T.$
    Thus, with high probability after a large enough
    number of rounds, for every agent $i \in [n]$ the 
    algorithm is bidding the reported value
    $\hat v_i$ no matter what the other reports $\hat v_{-i}$ are. Since the auction $A'$ is truthful,
    the utility of each agent is maximized
    when $b_i^T = v_i.$ Hence, the optimal
    strategy, up to $o(1)$, is to report
    $v_i = \hat v_i, \forall i \in [n].$
    To be more formal,
    the expected utility of the $i-$th agent
    in round $T$ is 
    \begin{align*}
        \E\left[u'_i(b_i^T,b_{-i}^T)\right] &= u'_i(\hat v_i, \hat v_{-i}) + o(1) \,,
    \end{align*}
    thus, since $A'$ is truthful, this quantity is maximized
    for $\hat v_i = v_i,$ up to the $o(1)$ term.
\end{proof}

\section{Omitted Details from \Cref{sec:no-regret-learning-auctioneer}}
\label{app:finite-horizon}

\begin{proof}[Proof of \Cref{prop:regret-single-auction-upper-bound}]
Let $A_T = p_T \cdot A + (1-p_T)\cdot \mathrm{SPA}, $ where $A$ is some auction
with $\gamma_A > 0$ and some $p_T$ that will be defined shortly. Notice that
\[
    \gamma_{A_T} \geq p_T \cdot \gamma_A + (1-p_T) \cdot \gamma_{\mathrm{SPA}} \geq p_T \cdot \gamma_A \,. 
\]
Since the bidders are mean-based no-regret learners, we know that when
\[
\sum_{\tau = 1}^{T_0}v_i \cdot x_i(v_i,b_\tau) - p_i(v_i, b_\tau) \geq \sum_{\tau = 1}^{T_0}v_i \cdot x_i(b',b_\tau) - p_i(b',b_\tau) + \delta_T \cdot T,  \forall i \in \{0,1\},\forall b' \in B_\Delta \,, 
\]
they will be bidding truthfully with probability at least $1-\Delta\cdot \eta_T.$
We know that in every round 
\begin{align*}
    v_i \cdot x_i(v_i,b_\tau) - p_i(v_i,b_\tau) &\geq v_i \cdot x_i(b',b_\tau) - p_i(b',b_\tau) + \gamma_{A_T}\\
    &\geq v_i \cdot x_i(b',b_\tau) - p_i(b',b_\tau) + p_T \cdot \gamma_{A}, \forall i \in \{0,1\}, b_\tau, b' \in B_\Delta^2, b'\neq v_i
\end{align*}
Thus, we define $T_0 = \min\{t \in \nats: p_T \cdot \gamma_A \cdot t \geq \delta_T \cdot T   \} = \nicefrac{\delta_T \cdot T}{p_T \cdot \gamma_A}.$ The regret is
\begin{align*}
    \wt{\reg}_T(A_T;v_L,v_H) &= \wt{\reg}_{T_0}(A_T;v_L,v_H) + \left( \sum_{t=1}^T \rev(v_L,v_H;\SP) - \E\left[\sum_{t=T_0+1}^T \rev(b_L^t,b_H^t;A) \right]\right) \\
    &\leq v_L \cdot T_0 + v_L\cdot(T-T_0)\cdot(2\Delta \cdot \delta_T)\cdot(1-p_T) + (T-T_0)\cdot p_T \cdot v_L\\
    &\leq v_L \cdot \left(T_0 + 2\Delta \cdot \delta_T\cdot T\cdot(1-p_T) + T\cdot p_T\right) \\
    &\leq v_L \cdot \left(\frac{\delta_T \cdot T}{p_T\cdot \gamma_A} + 2\Delta\cdot \delta_T \cdot T + p_T \cdot T \right) \\
    &\leq v_L \cdot \left(\frac{2\Delta \cdot \delta_T \cdot T}{p_T\cdot \gamma_A}  + p_T \cdot T \right) \,,
\end{align*}
where the first inequality follows from the fact that
after the first $T_0$ rounds the auctioneer regret is bounded the sum
of the probabilities that the auction is SPA and the bidders do not bid
truthfully, which is at most $(1-p)\cdot 2\Delta \cdot \eta_T,$
and the probability that auction is not SPA, which is $p_T.$
The rest of the inequalities are just algebraic manipulations.
Thus, by setting $p_T = \sqrt{\nicefrac{2\Delta \cdot \delta_T}{\gamma_A}}$
we get that 
\[
     \wt{\reg}_T(A_T;v_L,v_H) \leq v_L \cdot \left(3\cdot \sqrt{\frac{2\Delta \cdot \delta_T}{\gamma_A}} \cdot T \right) \,,
\]
which concludes the proof.
\end{proof}

\begin{proof}[Proof of \Cref{prop:lower-bound-constant-auctions}]
Consider the $v_L, v_H$ pairs of the form $v_H = v_L + \nf1\Delta$, such that both are bounded away from $0$ and $1$. Then, Myerson's payment formula shows that $p_H(v_H, v_L) \le (v_H - \nf1\Delta) \cdot x_H(v_H,v_L) = v_L \cdot x_H(v_H,v_L).$ We first argue that $x_H(v_H,v_L) < 1.$ Indeed, suppose that $x_H(v_H,v_L) = 1.$ Then the low type
gets no signal about their bid and hence bids uniformly 
at random between $[0,v_L]$.
%\footnote{As we showed in \Cref{theorem:deterministic anonymous auctions do not lead to convergence},
%this holds even if the high type
%is a learning agent; now the distribution of the low type is bidding from a different distribution,
%but it is still not truthful.} 
In particular, with some $C_\Delta$ probability
that is independent of $T,$ the low type bids the value $b_L = v_L/2$.
Now the only way
for the auction $A_T$ to generate $(v_L-o(1))$ revenue from such rounds is if $x_H(v_H, v_L/2) - x_H(v_L,v_L/2) = 1-o(1).$ But if this is the case, then consider the valuation pair $(v_L/2, v_L/2+\nf1\Delta)$: the auctioneer allocates at most $x_H(v_L/2+\nf1\Delta, v_L/2) \leq o(1)$ per round, and gets almost no revenue from the high type. Moreover, the low type will generate at most $v_L/2$ revenue, so the the regret of the auctioneer
is at linear in $T$; this gives the desired contradiction.

Since $x_H(v_H,v_L) < 1$, let $q := 1-x_H(v_H,v_L)$. Then, $x_L(v_L,v_H) \leq q$ and so
$u_L(v_L,v_H) - u_L(v_L-\nicefrac{1}{\Delta}, v_H) \leq q \cdot \nf1\Delta \leq q$.
In order to cancel the effect
of the learning rate of $\eta_T$, we need to wait for $T_0 := \nicefrac{\Omega(1)}{(q\cdot \eta_T)}$ rounds. %\anupamnote{Say more.}
For some $C'_\Delta$ fraction of these $T_0$ rounds the agent of low type will bid
$v_L/2$, and an argument similar the previous paragraph shows that %\alert{as we argued before when this happens}
the revenue of the auction will be
at least $1/\Delta - o(1)$ less than $v_L$. Thus, the regret in these $T_0$ rounds
will be $\Omega(T_0),$ where we are hiding  constants depending on $\Delta$.
Let us assume that after $T_0$ rounds the low type starts bidding truthfully.
Then, the total regret in this period due to allocation of the item to the low type is $\Omega\left((T-T_0)\cdot q\right)$.
Summing up the two terms we get a regret of $\Omega\left(\nf{1}{(q\eta_T)} + q\cdot T - \nf{1}{\eta_T} \right).$
Since $\eta_T = \Theta(1/\sqrt{T})$, this is $\Omega(\sqrt{T}/q + qT - \sqrt{T})$, which for any choice of $q$
 is $\Omega\left(T^{3/4}\right).$
\end{proof}

% \section{Missing Proofs}
% \label{app:missing}

\begin{proof}[Proof of \Cref{thm:auction-schedule-regret-upper-bound}]
We will upper bound the auctioneer regret in the two epochs $\{1,\ldots,T_0\},$
and $\{T_0 + 1,\ldots, T\},$ separately, where $T_0 \in [T]$
is a parameter of the design which we will define shortly.
For the first epoch, we will
use the simple upper bound of $v_L \cdot T_0.$ 

Let us consider the bid distribution of the two bidders
after $T_0$ rounds. Since they are mean-based no-regret learners
we know that if 
\[
    \sum_{\tau = 1}^{T_0}v_i \cdot x_i(v_i,b_\tau) - p_i(v_i, b_\tau) \geq \sum_{\tau = 1}^{T_0}v_i \cdot x_i(b',b_\tau) - p_i(b',b_\tau) + \delta_T \cdot T,  \forall i \in \{1,2\},\forall b' \in B_\Delta \,, 
\]
then, by a union bound over the possible bids, they will both be 
bidding truthfully with probability at least $1-2\Delta \cdot \eta_T.$

We know that in every round $\tau \in [T_0]$ we have that
\begin{align*}
    v_i \cdot x_i(v_i,b_\tau) - p_i(v_i,b_\tau) &\geq v_i \cdot x_i(b',b_\tau) - p_i(b',b_\tau) + \gamma_{A}, \forall i \in \{0,1\}, b_\tau, b' \in B_\Delta^2, b'\neq v_i \,.
\end{align*}

Therefore, we set $T_0 = \min\{ t \in \nats: t\cdot \gamma_A \cdot t \geq 
\delta_T \cdot T\} = \nicefrac{\delta_T \cdot T}{\gamma_A}.$
Thus, we can upper bound the cumulative auctioneer regret by
\begin{align*}
    \wt{\reg}(A,\ldots,A,\mathrm{SPA},\ldots,\mathrm{SPA};v_L,v_H) &\leq v_L \cdot T_0 + v_L \cdot (T - T_0)\cdot 2\Delta \cdot \eta_T \\
    &\leq v_L \cdot \frac{\delta_T \cdot T}{\gamma_A} + v_L \cdot T \cdot 2\Delta \cdot \eta_T \\
    &= O\left(  \delta_T \cdot T \cdot \left(\frac{1}{\gamma_A} + \Delta \right)\right) \,,
\end{align*}
where the first inequality follows from the fact that with probability
at most $2\Delta\cdot \eta_T$ one of the two bidders will not be truthful
in the last $(T-T_0)$ rounds, and the other inequalities are just
algebraic manipulations.
\end{proof}

\begin{proof}[Proof of \Cref{prop:root-T-lower-bound-all-environments}]
It is not hard to see that in the setting we are working on 
the auctioneer cannot have negative auctioneer regret in any interval of the interaction. 
For instance, when $v_H = v_L - 1/\Delta,$ the SPA performs optimally. 
Since every $A_t, t \in [T],$ is a truthful auction, Myerson's 
lemma shows that 
\[
    u^t_i(v_i,b_{-i}) - u^t_i(b',b_{-i}) = \int_{z = b'}^{v_i} x^t_i(z,b_{-i}) dz - \left(v_i - b' \right) \cdot x^t_i(b', b_{-i}), \forall i \in \{1,2\}, \forall v_i, b',b_{-i} \in B^3_\Delta \,,
\]
so for $b' = v_i - 1/\Delta$ we get that
\[
    u^t_i(v_i,b_{-i}) - u^t_i(v_i - 1/\Delta,b_{-i}) \leq \frac{1}{\Delta} ,\forall v_i, b',b_{-i} \in B^3_\Delta \,.
\]
Thus, in every iteration the utility gain of bidding $v_i$ is at most $1/\Delta$
greater than bidding $v_i - 1/\Delta.$ Summing up over the first $T_0$ iterations, we get
that 
\[
    \sum_{t=1}^{T_0} \left( u^t_i(v_i,b_{-i}) - u^t_i(v_i - 1/\Delta,b_{-i})\right) \leq \frac{T_0}{\Delta} ,\forall v_i, b',b_{-i} \in B^3_\Delta \,.
\]
Let us now shift our attention to the weights that MWU puts on $v_i-1/\Delta, v_i,$
after $T_0$ iterations.
We have
\begin{align*}
    \frac{\Pr[b^{T_0}_i = v_i]}{\Pr[b^{T_0} = v_i-1/\Delta]} &= e^{\eta_T \sum_{t=1}^{T_0}\left(  u^t_i(v_i,b^t_{-i}) - u^t_i(v_i - 1/\Delta,b^{t}_{-i})\right)}\\
    &\leq e^{\eta_T \cdot \frac{T_0}{\Delta}} \,,
\end{align*}
so for $T_0 = \nicefrac{\Delta}{\eta_T}$ we have that
\[
    \Pr[b^{T_0} = v_i-1/\Delta] \geq \frac{\Pr[b^{T_0}_i = v_i]}{e} \,.
\]
This immediately implies that
\[
    \Pr[b^{t} = v_i-1/\Delta] \geq \frac{\Pr[b^{t}_i = v_i]}{e}, \forall t \in [T_0] \,.
\]
Thus, the probability of bidding truthfully of both algorithms
is bounded by $9/10.$ Thus, when $v_H = v_L + 1/\Delta$ when both bidders
are not bidding truthfully the revenue loss compared to $\mathrm{SPA}$
is at least $1/\Delta.$ Putting it together, we can see that
within the first $T_0$ rounds the total revenue loss compared to $\mathrm{SPA}$ is at least $C\cdot \nicefrac{1}{\Delta} \cdot T_0 = C\cdot \eta_T = C\cdot \sqrt{T},$ for some absolute constant $C>0.$
\end{proof}

Next, we show that the auction we defined in \Cref{def:staircase-auction} 
is optimal, in terms of its parameter $\gamma_A.$ 

\begin{lemma}\label{lem:optimal-gamma}
In the setting with two bidders it holds
that the optimal choice of the parameter $\gamma_A$
is
$\Theta\left(\nicefrac{1}{\Delta^2}\right).$
Moreover, the auction defined in \Cref{def:staircase-auction} achieves
that bound.
\end{lemma}

\begin{proof}[Proof of \Cref{lem:optimal-gamma}]
Consider some auction $A$ and fix the bid of the second bidder to be $b' \in B_\Delta.$ Then,
$x_1(\cdot, b')$ is a non-decreasing function, 
with $0 \leq x_1(b,b') \leq 1,\forall b \in B_\Delta.$ Notice that
for any consecutive bids, Myerson's lemma shows that
\[
    u_1(b,b') - u_1(b-1/\Delta,b') \leq \nicefrac{1}{\Delta}\cdot \left(x_1(b,b') - x_1(b-1/\Delta,b') \right) \,.
\]
Since there are $1/\Delta$ different $b \in B_\Delta$ and the function 
$x_1(\cdot,b')$ is monotone and bounded between $[0,1]$ we have
\begin{align*}
    \sum_{b > 0} x_1(b,b') - x_1(b-1/\Delta,b') &= x_1(1,b') - x_1(0,b') \\
    &\leq 1\,,
\end{align*}
and since there are $1/\Delta$ terms in the summation, all of which are non-negative at least one of them must be at most $1/\Delta.$ Let
$b^*_1 \in B_\Delta$ be such that $x(b^*_1,b') - x(b^*-1/\Delta,b') \leq \frac{1}{\Delta}.$ Then, picking $v_1 = b^*_1$ witnesses that 
$\gamma_A \leq \frac{1}{\Delta^2}.$

% Going back to the auction from \Cref{def:staircase-auction}, we see that
% for any $i \in \{1,2\}, b_i, b_{-i} \in B^2_\Delta$ it holds that 
% \[
%     x_i(b_i, b_{-i}) - x_i(b_i-1/\Delta,b_{-i}) = 2/\Delta \,,
% \]
% which means that
% \[
%     u_i(b_i, b_{-i}) - u_i(b_i-1/\Delta,b_{-i}) = 2/\Delta^2 \,.
% \]
\end{proof}

\section{Extensions}\label{app:extensions}
In this section we discuss
potential extensions of our model
and adaptations of our results.

\paragraph{Extension to partial feedback setting.}
Our results can be adapted to the partial feedback
setting, with different quantitative bounds. In particular,
there are mean-based no-regret algorithms such as EXP3 \citep{braverman2018selling} with $\eta_T = \widetilde{O}(T^{1/4}).$ Notice that our positive results are stated for mean-based learners, so the guarantees hold in this setting as well.

\paragraph{Extension to multiple bidders.} 
We underline that
        our results in \Cref{sec:randomized-truthful-auctions} are already stated
        and proven for multiple bidders. For our upper
        bounds in \Cref{sec:no-regret-learning-auctioneer} there is a $1/n$ degradation to the auctioneer regret bound.
        When we are dealing with $n$ bidders 
        we can create
        a strictly IC auction $A$ by building upon
        our ``staircase auction''
        approach for two bidders in the following way:
        we select some bidder $i \in [n]$ uniformly at random
        (independently of their bids) and then we allocate to
        bidder $i$ with probability $b_i.$ Thus, for
        each bidder $i \in [n]$ their allocation probability $x_i(b)$
        is a linear function with $x_i(0) = 0, x_i(1) = 1/n.$
        Hence, Myerson's lemma shows that $u_i(v_i) - u_i(v_i - 1/\Delta) = \Theta(1/(n\Delta^2)),$ thus, $\gamma_A = \Theta(1/(n\Delta^2)).$ Recall that in the two-bidder
        case we have shown that this auction gives $\gamma_A = \Theta(1/\Delta^2),$ so the degradation in $\gamma_A$
        by $1/n$ leads to a degradation of the same factor
        in the auctioneer regret compared to the two-bidder setting.

\paragraph{Extension of regret bounds
to the distributional setting.}In \Cref{sec:no-regret-learning-auctioneer} we consider a setting where the auctioneer
        does not have any distributional knowledge
        about the valuation of the bidders. Notice
        that our lower bounds are witnessed
        by valuation pairs of the low type, high type,
        of the form $v_L = v, v_H = v+1/\Delta.$ Let us now 
        consider a distributional setting where $v_1, v_2$ are drawn
        from distributions $\mathcal{D}_1, \mathcal{D}_2,$
        and then the two bidders participate in repeated
        second-price auctions using MWU parametrized by these valuations. Similarly as in the prior-free setting, 
        the goal of the auctioneer is to have
        small expected regret, where the expectation
        is over the random draw of the valuations
        and the random behavior of MWU. Notice that the cumulative
        revenue of SPA when the bidders are truthful is $T \cdot \mathbb{E}_{v_1 \sim \mathcal{D}_1, v_2 \sim \mathcal{D}_2}[\min\{v_1,v_2\}],$ so this is the benchmark the auctioneer
        competes with (in this setting, 
        we can modify the benchmark to be SPA with
        personalized reserves with the same arguments).
        If these distributions $\mathcal{D}_1, \mathcal{D}_2,$ 
        place some constant probability (i.e.,
        independent of $T$)
        on every element of $\{0,1/\Delta,2/\Delta,\ldots,1\}$
        then with some constant probability we will see
        a draw of the form $v_L = v, v_H = v+1/\Delta$, so
        these pairs will be contributing a constant fraction
        of the expected revenue of the second-price auction, i.e.,
        the term $\mathbb{E}_{v_1 \sim \mathcal{D}_1, v_2 \sim \mathcal{D}_2}[\min\{v_1,v_2\}].$ Thus,
        if the auctioneer wants to have expected regret
        at most $O(R_T)$, they need to have regret at most
        $O(R_T)$ for all such valuation pairs, where in the
        notation $O(\cdot)$ we are suppressing all the parameters
        that do not depend on $T.$

% \input{omitted-details-randomized}
%%%%%%%%%%%%%%%%%%%%%%%%%%%%%%%%%%%%%%%%%%%%%%%%%%%%%%%%%%%%

\end{document}